\documentclass{article}
\usepackage{lmodern}

 \usepackage[main,final]{neurips_2026}

\usepackage[utf8]{inputenc} 
\usepackage[T1]{fontenc}    
\usepackage[hidelinks]{hyperref}       
\hypersetup{pdftitle={SecureClaw: Clawing Back Control of LLM Agents},pdfauthor={Yuhan Ma and Stefan Schmid}}
\usepackage{url}            
\usepackage{booktabs}       
\usepackage{amsmath}        
\usepackage{amsfonts}       
\usepackage{amsthm}         
\IfFileExists{nicefrac.sty}{\usepackage{nicefrac}}{}       
\IfFileExists{microtype.sty}{\usepackage[protrusion=true,expansion=false]{microtype}}{}      
\IfFileExists{xcolor.sty}{\usepackage{xcolor}}{}         
\usepackage{tikz}           
\usetikzlibrary{positioning,arrows.meta,shapes,calc,patterns,backgrounds,fit}
\IfFileExists{pgfplots.sty}{\usepackage{pgfplots}\pgfplotsset{compat=1.18}\newif\ifpgfplotsloaded\pgfplotsloadedtrue}{\newif\ifpgfplotsloaded\pgfplotsloadedfalse}  
\usepackage{pifont}         
\usepackage{graphicx}       
\IfFileExists{enumitem.sty}{\usepackage{enumitem}}{}       
\IfFileExists{multirow.sty}{\usepackage{multirow}}{}       

\newtheorem{theorem}{Theorem}

\newtheorem{proposition}[theorem]{Proposition}
\usepackage{booktabs}
\usepackage{tabularx}
\usepackage{array}
\usepackage{ragged2e}
\usepackage{algorithm}
\usepackage{makecell}
\usepackage{algpseudocode}
\usepackage{placeins}  
\usepackage{subcaption}
\floatname{algorithm}{Protocol}
\newcolumntype{L}[1]{>{\RaggedRight\arraybackslash}p{#1}}
\newcolumntype{C}[1]{>{\Centering\arraybackslash}p{#1}}
\newcolumntype{Y}{>{\RaggedRight\arraybackslash}X}
\title{SecureClaw: Clawing Back Control\\ of LLM Agents}

\author{%
  Yuhan Ma\textsuperscript{1}\thanks{Corresponding authors.}\quad
  Stefan Schmid\textsuperscript{1}\footnotemark[1]\\
  \textsuperscript{1}TU Berlin\\
  \texttt{yuhan.ma@campus.tu-berlin.de}\quad
  \texttt{stefan.schmid@tu-berlin.de}%
}

\begin{document}

\maketitle

\begin{abstract}
Tool-using large language model (LLM) agents face two distinct security failures: unauthorized external actions and exposure of sensitive plaintext inside the runtime before any final output check can intervene. Existing defenses usually protect one boundary, either the planner/runtime or the point where external actions are executed, and therefore do not by themselves secure both surfaces. We present SecureClaw, a dual-boundary architecture that enforces authorization in a trusted executor and keeps raw protected data outside the LLM runtime. Sensitive reads pass through a trusted gateway that replaces raw values with reference IDs. In the evaluated deployment, the gateway also provides short planning summaries that deliberately release selected information, a process called declassification. Writes that change external state follow a PREVIEW$\rightarrow$COMMIT protocol in which only a trusted executor may execute the exact request authorized by policy. The runtime can still plan over summaries and reference IDs, but cannot directly resolve the IDs into secrets or perform side effects. Across the AgentDojo, AgentLeak, and Agent Security Bench (ASB) benchmarks, SecureClaw is the only defense we evaluate in a common harness that simultaneously retains usable task utility and achieves 0\% attack success rate (ASR) on ASB, 0.64\% ASR on AgentDojo, and 3.23\% overall leak on AgentLeak's attacked parity lane, which measures final-output and internal-relay leakage. Code is available at \url{https://github.com/ponyyuhan/SecureClaw}.
\end{abstract}

\section{Introduction}
\label{sec:intro}

Tool-using language-model agents are moving from chat-style assistance to workflow automation. A single agent may read email threads, invoices, customer-support tickets, calendars, code repositories, or medical intake forms, and then act on the outside world by sending messages, sharing files, updating records, or scheduling meetings. 

In these settings, a failure is not merely a poor answer: it can become an irreversible external effect or an exposure of sensitive data during execution.

This threat model is well established~\citep{yu2025survey}. Indirect prompt injection showed that attacker-controlled content can steer LLM-integrated systems by collapsing the boundary between data and instructions~\citep{greshake2023not}. AgentDojo and Agent Security Bench (ASB) show that tool-using agents remain vulnerable to prompt injection that redirects actions toward attacker-chosen outcomes~\citep{debenedetti2024agentdojo,zhang2025agent}. More recent work broadens the attack surface to tool selection, web agents, and agent-mediated financial workflows~\citep{toolselect2025,wasp2025}. In parallel, AgentLeak shows that multi-agent systems introduce internal leakage paths that output-only auditing does not capture, especially inter-agent messages and shared memory~\citep{agentleak2026}.

These results point to two critical security questions. The first is effect authorization: whether the executed action matches the exact request authorized for the current caller, session, inputs, and context. The second is protection of raw data: whether the untrusted runtime ever receives sensitive data in usable form. These questions are related but not interchangeable. A boundary around effectful tools can stop an unauthorized commit, yet it does not help if the runtime has already read a secret and can relay it through internal channels. Conversely, runtime-local information-flow defenses can reduce some leakage, but they still leave substantial trust in planner state or runtime enforcement. Architectural and execution-time systems move control closer to the action or resource boundary~\citep{rtbas2025,costa2025securing,camel2025,ace2025,progent2025,pcas2026,ac4a2026}.

SecureClaw keeps control of external actions and raw protected values outside the untrusted agent runtime. When a tool returns protected data, a trusted gateway stores the raw value and gives the runtime a reference ID plus a short planning summary. When the runtime proposes an external action, the policy engine checks the complete request. A trusted executor independently verifies the approval before retrieving the referenced data and calling the external API.
\paragraph{\emph{Running example.}}
Consider an agent that reads an invoice and sends it to an approver. Figure~\ref{fig:architecture} follows four steps.
\textbf{(1) Read.} The gateway stores the raw invoice and returns a reference ID, written as \texttt{INV\_42} for readability, plus a summary such as ``vendor A, amount \$4.2k, due Friday.''
\textbf{(2) Plan.} The runtime reasons over this summary and carries the ID through memory and later calls; it cannot use the ID to retrieve the raw invoice.
\textbf{(3) Preview.} The runtime proposes sending \texttt{INV\_42} to \texttt{approver@corp.com}. The gateway adds the authenticated caller and session and serializes the recipient, channel, referenced object and version, and relevant context in a fixed form. It then submits the complete request to the policy engine, which binds its approval to exactly that request.
\textbf{(4) Commit.} The executor, the only component that can call the email API, checks that the approval matches the request, is fresh, and has not been reused; only then does it resolve the reference and send. Changing the recipient after approval causes rejection; a denial can return a safe hint such as ``save a draft and request approval.''

In this architecture, the runtime still plans, searches, and composes tool calls, but it does so with reference IDs rather than unrestricted raw data, and it cannot directly turn a proposal into a committed external action. The reference-ID mechanism follows a capability-style design in which authority is carried by an unforgeable reference rather than by exposure to raw data~\citep{saltzer1975protection,watson2010capsicum}. When the runtime needs task-relevant information, the summary interface releases only a short, policy-bounded summary $D(v)$ of a sensitive value $v$. This is an explicit declassification channel: an intentional, auditable release of selected information from a protected value~\citep{sabelfeld2005declassification}.

The two risks arise at different points: an external effect occurs when the API executes the action, whereas sensitive data is exposed as soon as a raw protected value enters an adversarial runtime. This predicts a concrete empirical pattern: executor-side enforcement should stop unauthorized commits without closing internal relay channels, while keeping raw data outside the runtime should close relay channels without authorizing commits. The ablation and bypass suite test whether either control can replace the other.

This paper makes four contributions.

\begin{itemize}
    \item \textbf{A security decomposition for agent systems.}
    We identify two non-substitutable security requirements for tool-using agents: request-bound authorization of external effects and confinement of sensitive plaintext away from the untrusted runtime.

    \item \textbf{A practical dual-boundary architecture.}
    We present SecureClaw, which mediates sensitive reads with reference IDs, effectful writes with a PREVIEW$\rightarrow$COMMIT executor, and blocked actions with deny-aware recovery: a fixed-template safe-continuation path after a denied commit. In the evaluated deployment, short planning summaries form an explicit declassification interface.
    
    \item \textbf{A common-harness evaluation that separates the two surfaces.}
    Across AgentDojo, AgentLeak, and ASB, SecureClaw is the only evaluated common-harness baseline that simultaneously reaches 0\% ASR on ASB, 0.64\% ASR on AgentDojo, and 3.23\% overall leak on AgentLeak's attacked parity lane; the ablation and bypass suite then isolate why executor-side authorization and read-side confinement are complementary rather than interchangeable.
    
    \item \textbf{A reusable implementation artifact.}
    We provide the SecureClaw runtime implementation and benchmark evaluation code at \url{https://github.com/ponyyuhan/SecureClaw}, with an archived version at \url{https://doi.org/10.5281/zenodo.23258171} (Appendix~\ref{app:code-availability}).
    
\end{itemize}

\FloatBarrier

\section{Problem setting and scope}
\label{sec:problem}

We consider a tool-using agent deployment with an LLM-driven \textbf{runtime} that plans and invokes tools; a trusted \textbf{gateway} on the read path; a trusted \textbf{data store} for protected values; a \textbf{policy engine} that authorizes previewed requests; a trusted \textbf{executor} on the write path; and a \textbf{user} who may confirm high-risk actions. The policy engine may be a single service, with an optional distributed realization for deployments that also want transcript privacy against one evaluator (Appendix~\ref{app:proofs:policy}). An action is \textbf{effectful} if it can change external state or send information to an external party, for example by sending an email, sharing a file, updating a record, or initiating a payment. A value is \textbf{sensitive} if exposing it to the runtime in plaintext would create a meaningful disclosure risk.

\begin{table}[t]
\centering
\caption{Two security surfaces, their control points, and their validation.}
\label{tab:surface-summary}
\small
\setlength{\tabcolsep}{5pt}
\renewcommand{\arraystretch}{1.18}
\begin{tabularx}{\linewidth}{
>{\bfseries\RaggedRight\arraybackslash}p{0.16\linewidth}
Y
Y
>{\RaggedRight\arraybackslash}p{0.16\linewidth}
}
\toprule
Surface & If left inside the runtime & SecureClaw control point & Main evidence \\
\midrule
Effect surface
& A compromised or injected runtime can mutate recipients, channels, or context and still reach the API that executes the action.
& Executor-side PREVIEW$\rightarrow$COMMIT, request binding, freshness checks, and replay protection.
& AgentDojo, ASB, bypass suite \\
\midrule
Plaintext surface
& A runtime that reads plaintext can relay it through internal channels such as inter-agent messages (C2) or memory (C5) even when a later commit is blocked.
& Gateway-side reference IDs, resolution by trusted components for authorized actions, and explicit short planning summaries, so that protected values do not re-enter the runtime outside declared declassification interfaces.
& AgentLeak, ablation \\
\bottomrule
\end{tabularx}
\end{table}

SecureClaw targets two security surfaces. The \textbf{effect surface} asks whether the executed action matches the exact request authorized for the current caller, session, inputs, and context. The \textbf{plaintext surface} asks whether the untrusted runtime ever receives raw secrets that it can relay through user-visible outputs, inter-agent messages, memory, logs, or later tool calls. AgentLeak's channel taxonomy is useful here: C1 denotes the final response, C2 inter-agent messages, C3/C4 tool I/O, C5 shared or persistent memory, C6 logs, and C7 persistent artifacts~\citep{agentleak2026}. For SecureClaw, C2 and C5 are especially important because they expose relay paths that remain invisible to output-only defenses~\citep{kong2025survey}.

We evaluate two adversaries. In the \textbf{prompt-injection} setting, attacker-controlled instructions appear inside retrieved content and attempt to redirect downstream actions or induce disclosure, as in AgentDojo, AgentLeak, and ASB~\citep{debenedetti2024agentdojo,agentleak2026,zhang2025agent}. In the \textbf{compromised-runtime} setting, the runtime is fully adversarial: it may fabricate requests, tamper with recipients or channels, replay prior authorizations, mix sessions, guess reference IDs, or exploit implementation mistakes. The core boundary guarantees trust the gateway, the trusted data store, the policy engine, and the executor. Appendix~\ref{app:proofs:policy} states the extra assumptions for the optional distributed policy variant. The design further assumes complete mediation of APIs that execute external actions and sensitive reads, correct protected-field classification, bounded clock drift, and crash-safe replay protection~\citep{saltzer1975protection}. 

Our goals are request-bound effect authorization, confinement of raw sensitive plaintext outside declared declassification interfaces, and explicit bounded declassification through the summary interface. SecureClaw does not claim that every policy-allowed action is semantically correct for the user's latent intent, nor does it survive compromise of the gateway or executor. By \textbf{residual failures}, we mean failures that remain after the two target boundaries behave as specified, for example a policy-allowed action with the wrong principal or object because the policy is too coarse. We analyze those residuals explicitly in Section~\ref{sec:exp:residual}.
 
Table~\ref{tab:surface-summary} summarizes the two security surfaces, why they require different boundary placement, and how each one is validated empirically.

\section{SecureClaw design}
\label{sec:design}

SecureClaw mediates effectful writes through a trusted executor and sensitive reads through a trusted gateway. The runtime remains useful because it continues to plan, search, route reference IDs, and compose tool calls, but it is no longer entrusted with plaintext or with committing external effects.

\begin{figure}[t]
\centering
\begin{subfigure}{\linewidth}
    \centering
    \includegraphics[width=\linewidth]{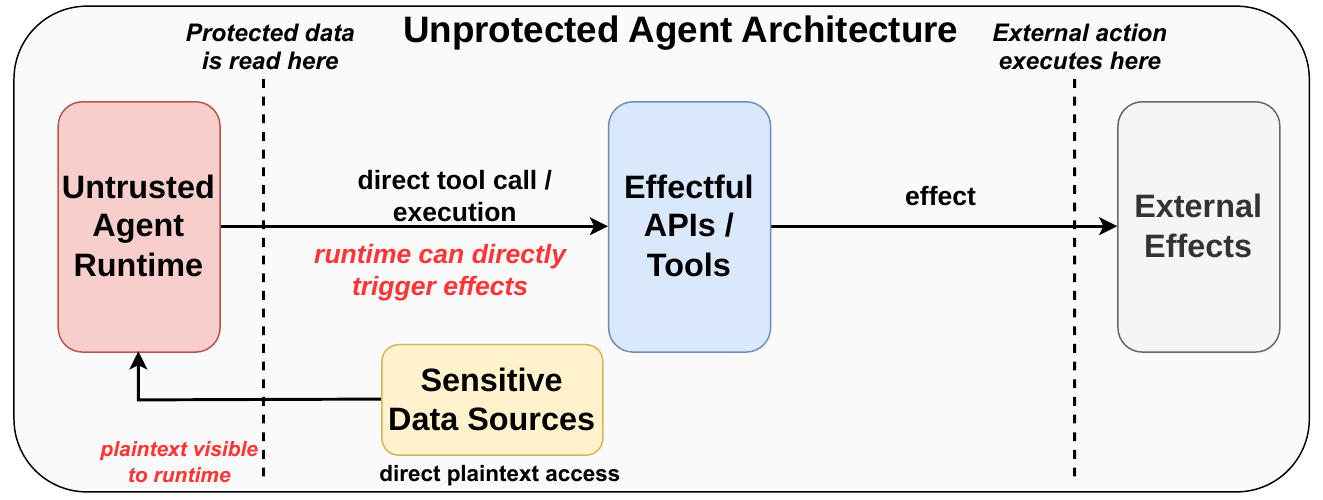}
    \caption{Unprotected agent architecture.}
    \label{fig:secureclaw-before}
\end{subfigure}

\smallskip

\begin{subfigure}{\linewidth}
    \centering
    \includegraphics[width=\linewidth]{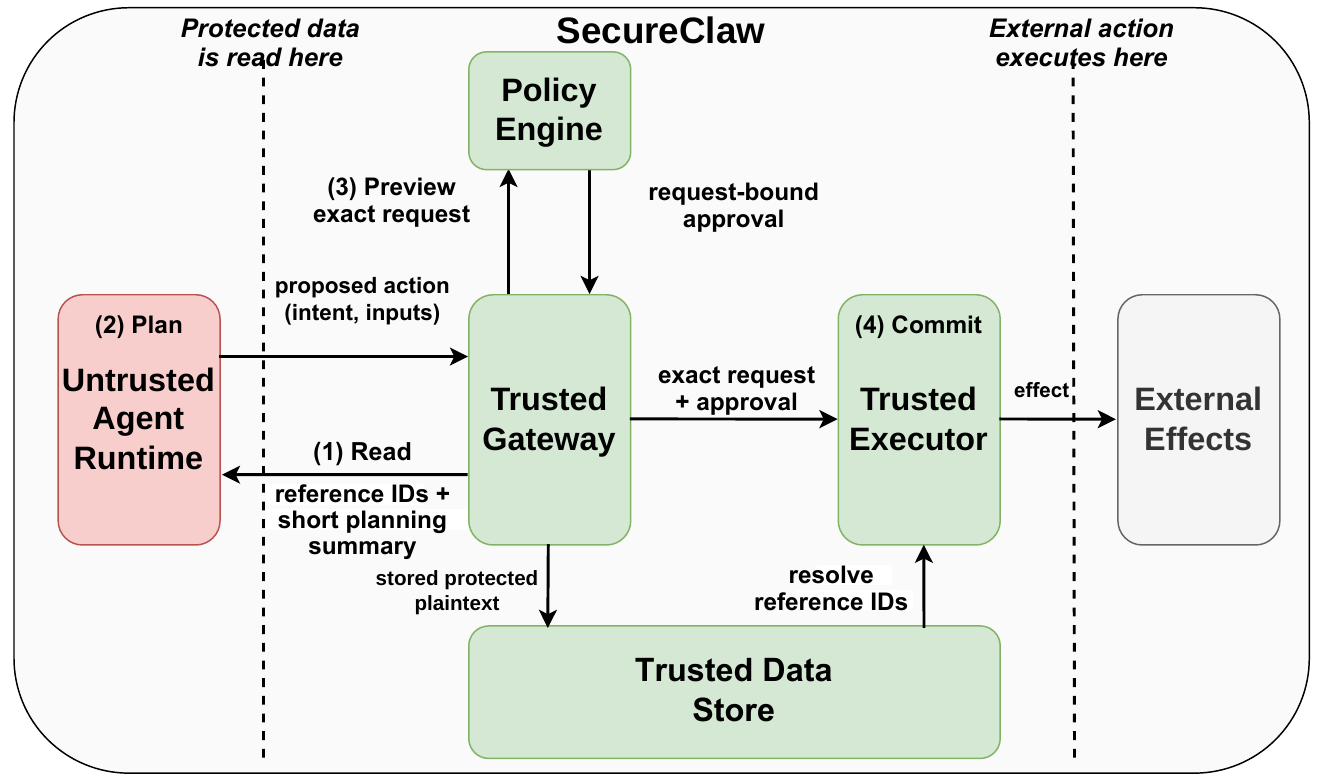}
    \caption{SecureClaw architecture.}
    \label{fig:secureclaw-after}
\end{subfigure}

\caption{Comparison between an unprotected agent architecture and SecureClaw. (1) \emph{Read}: the gateway returns reference IDs and short planning summaries. (2) \emph{Plan}: the runtime uses this view to propose an action. (3) \emph{Preview}: the gateway assembles the exact request and the policy engine binds its approval to it. (4) \emph{Commit}: the executor verifies that approval before resolving references and performing the external action.}
\label{fig:architecture}
\end{figure}

\subsection{Design overview and component roles}
\label{sec:design:overview}

SecureClaw separates planning from security enforcement. The \textbf{runtime} reasons, searches, and orchestrates tool calls, but it is untrusted for both authorization and confidentiality. The \textbf{gateway} mediates sensitive reads, classifies returned data, creates reference IDs, and canonicalizes effectful preview requests by serializing action-relevant fields in a fixed form. The \textbf{trusted data store} retains protected plaintext and reference-ID metadata outside runtime memory. The \textbf{policy engine} checks the complete request against configured rules and returns an authorization artifact: a verifiable allow/deny decision bound to that request. The \textbf{executor} is the only component that can call APIs that perform external actions. Before execution, it recomputes the request binding and verifies the authorization artifact, checks that the approval is fresh and has not been reused, validates any required confirmation token, and enforces restrictions on resolving reference IDs.

An execution therefore has two mediated phases. On the read path, the gateway returns a reference ID and, when planning requires it, a short planning summary produced by a deterministic, schema-aware summary operator. On the write path, every effectful request remains a proposal until executor-side verification succeeds.

\subsection{Read path: reference IDs and trusted storage}
\label{sec:design:handles}
The read path protects the plaintext surface. When a tool returns sensitive content, the gateway stores the plaintext in the trusted data store and returns a \textbf{reference ID} to the runtime (an \emph{opaque handle}, as in FlowFence~\citep{fernandes2016flowfence}). The ID is an unpredictable reference to a value in trusted storage. Its metadata records the authenticated caller (who may use the ID), the current interaction (session), its expiry time (time-to-live), the object version, and the external destinations for which resolution is permitted. Resolving a reference means replacing the ID with its stored value. In practice, a reference ID can be generated from a fresh nonce (a value used only once) and a secret handle key, e.g., $\mathsf{hid}=\mathrm{HMAC\text{-}SHA256}(k_{\mathsf{handle}},\mathsf{nonce}\|\mathsf{session}\|\mathsf{object})$, and resolved by lookup in the trusted store. The ID refers to the stored value; it is neither a hash of the secret nor an encrypted copy. The runtime may route the ID through later computation, memory, or tool calls, but it cannot resolve the ID into plaintext.

The gateway returns a reference ID and a short planning summary. The summary operator's deterministic core uses schema rules to retain up to eight fields or items per container and truncate text fields at 512 characters, before adding truncation indicators. It replaces identifiers such as emails, account numbers, and URLs with per-read tokens and strips configured suspicious markers and tags, without calling an LLM. That summary is treated as an explicit declassification interface rather than an incidental copy of the original value~\citep{sabelfeld2005declassification}. The reference-ID layer therefore prevents a later use of an ID from returning raw protected values to the runtime: inter-agent messages and memory carry references rather than secrets retrieved from trusted storage~\citep{watson2010capsicum,saltzer1975protection}. The runtime can still use content deliberately released in the planning summary, but it cannot retrieve additional protected content by resolving an ID.

\FloatBarrier

\subsection{Write path: \texorpdfstring{PREVIEW$\rightarrow$COMMIT}{PREVIEW-to-COMMIT} execution boundary}
\label{sec:design:boundary}

The write path protects the effect surface. The runtime may request an effectful action, but it cannot commit that action directly. The runtime proposes only the action family and runtime-visible inputs. The trusted gateway or surrounding deployment context adds authenticated caller, session, and context fields relevant to authorizing and executing the action. It then canonicalizes the request by serializing all fields that determine the action in a fixed form. Thus every effectful request
\[
\rho = (\mathsf{intent},\mathsf{caller},\mathsf{session},\mathsf{inputs},\mathsf{ctx})
\]
passes through Protocol~\ref{prot:preview-commit}. Here $\mathsf{intent}$ is the action family, $\mathsf{caller}$ identifies the principal and $\mathsf{session}$ identifies the interaction, $\mathsf{inputs}$ contains tool arguments and reference IDs, and $\mathsf{ctx}$ contains other fields relevant to authorizing and executing the action, such as recipients, channels, object identifiers, provenance hashes, policy/schema versions, confirmation class, and object versions.

\begin{algorithm}[t]
\caption{PREVIEW$\rightarrow$COMMIT}
\label{prot:preview-commit}
\begin{algorithmic}[1]
\State The runtime proposes $(\mathsf{intent},\mathsf{inputs})$; the trusted gateway appends authenticated $\mathsf{caller}$, $\mathsf{session}$, and $\mathsf{ctx}$ to form $\rho$.
\State The gateway canonicalizes $\rho$ and computes the binding digest
\[
h \gets \mathrm{HMAC\text{-}SHA256}(k_{\mathrm{bind}}, \mathrm{Canon}(\rho)),
\]
where $k_{\mathrm{bind}}$ is a secret key shared only by the trusted binding/verification components.
\State The gateway queries the policy engine and obtains an authorization artifact that binds the allow/deny decision to the exact digest $h$.
\State The executor independently recomputes $h$, verifies the authorization artifact, checks freshness and replay state, validates any confirmation token, and enforces restrictions on resolving reference IDs.
\If{all checks succeed}
    \State The executor resolves any required reference IDs and commits the external effect.
\Else
    \State The executor rejects the request without executing the action (fails closed).
\EndIf
\end{algorithmic}
\end{algorithm}

The crucial property is that the executor is the only component that can reach the API that executes the external action. A compromised runtime may attempt to mutate the recipient, channel, or other fields after preview, but such changes alter the canonical request and therefore invalidate the bound authorization artifact. Theorem~\ref{thm:nbe} formalizes this as request-bound authorization integrity.

In the invoice example, $\mathrm{Canon}(\rho)$ serializes the recipient, channel, caller, session, and invoice version. Changing the recipient changes this serialization and invalidates the approval. The gateway and executor perform these checks outside the LLM runtime.

This distinction also clarifies why SecureClaw is stronger than argument filtering alone. If the runtime can still reach an API that executes an external action, checking arguments does not create a non-bypassable commit boundary. Zero-argument effectful tools make this explicit: there are no arguments to inspect. On the read side, once unrestricted plaintext enters an adversarial runtime, downstream checks cannot prevent it from being re-encoded into another channel. SecureClaw therefore separates control of external actions from access to raw protected data.

\subsection{Deny-aware recovery}
\label{sec:design:recovery}

Fail-closed execution is necessary for security, but a denial should not force the workflow to terminate when a safe alternative exists. When the executor rejects a commit, it returns a denial from a fixed set of templates. The response contains only a coarse action class, a reason code, and safe next-step hints; it does not reveal protected content obtained by resolving a reference. For example, it may suggest switching from ``send'' to ``draft for confirmation'' or requesting user approval. This recovery logic does not weaken the commit invariant; it preserves utility after blocked attacks.

\subsection{Threat-to-mechanism summary and formal scope}
\label{sec:formal}

Each SecureClaw mechanism protects against a different failure mode.
\paragraph{Unauthorized external effects.}
The execution boundary and request binding protect the effect surface: an external side effect can commit only if the executor verifies authorization for the exact canonicalized request that reaches the external API.

\paragraph{Internal plaintext relay.}
The reference-ID layer protects the plaintext surface: inter-agent messages, memory operations, and later tool calls carry reference IDs rather than raw sensitive values, except for content intentionally released through declared declassification interfaces. In the deployed system, any remaining runtime exposure occurs only through the short planning-summary interface.

\paragraph{Formal scope.}
The appendix mirrors this decomposition. Theorem~\ref{thm:nbe} formalizes request-bound authorization integrity for the commit path. Theorem~\ref{thm:sm} formalizes raw-plaintext confinement for the handle-only setting, without planning summaries. Theorem~\ref{thm:bounded-summary} then quantifies how much observing the planning summaries used in the experiments can help an adversarial runtime distinguish protected values. Appendix~\ref{app:proofs:policy} states an optional policy-side privacy refinement, and Appendix~\ref{app:proofs:compromise} summarizes the guarantees lost if trusted components are compromised.

\section{Experiments}
\label{sec:experiments}

We evaluate SecureClaw against our central claim: securing tool-using agents requires protecting two distinct surfaces, the effect surface and the plaintext surface.

The experiments ask whether SecureClaw (i) blocks unauthorized external effects, (ii) removes runtime-held plaintext from internal leakage channels, and (iii) keeps workflows usable after blocked actions without relaxing the commit invariant. We answer these questions with full-benchmark evaluations on AgentDojo, AgentLeak, and ASB, a mechanism ablation, a compromised-runtime bypass suite, and targeted measurements of deny-aware recovery and overhead.

\paragraph{Baselines and setup.}
  We compare against four common-harness baselines, meaning matched rows, model, scorer, temperature, and task protocol: \textbf{Plain} (no defense), \textbf{IPIGuard}~\citep{ipiguard2025}, \textbf{DRIFT}~\citep{drift2025}, and \textbf{Faramesh}~\citep{faramesh2026}. Appendix~\ref{app:methodology} specifies the code source, adapter, and defense surface for each baseline. Additional AgentDojo comparisons use the released Progent and CaMeL implementations through benchmark adapters (Appendix~\ref{app:additional-baselines}); Section~\ref{sec:related} compares their enforcement mechanisms. We report ASR, benchmark utility, and per-channel leak rates; SecureClaw's mediated outputs are scored after alias resolution and required confirmation handling. All main runs use \texttt{gpt-4o-mini-2024-07-18} with temperature 0 on AgentDojo~\citep{debenedetti2024agentdojo}, AgentLeak~\citep{agentleak2026}, and ASB~\citep{zhang2025agent}. The main-table AgentLeak columns use the 496 attacked scenarios in its parity lane, which checks leakage through final responses (C1), inter-agent messages (C2), and shared memory (C5). Appendix~\ref{app:methodology} gives exact counts, confidence intervals, and the original cross-model spot checks. Appendix~\ref{app:full-model-transfer} adds full-inventory SecureClaw transfer runs; these remain separate from the primary comparisons.

\subsection{Main cross-benchmark results}
\label{sec:exp:cross}

\begin{table}[t]
\centering
\caption{Benchmark evaluation on same-harness baselines (\texttt{gpt-4o-mini}). ASR = attack success rate; Atk.\ Util = benchmark utility on attacked rows; Benign = benign-task success rate.}
\label{tab:cross-bench}
\small
\setlength{\tabcolsep}{4pt}
\resizebox{\textwidth}{!}{
\begin{tabular}{l ccc ccc cc}
\toprule
& \multicolumn{3}{c}{\textbf{AgentDojo} ($n{=}629$)} & \multicolumn{3}{c}{\textbf{AgentLeak} ($n{=}496$)} & \multicolumn{2}{c}{\textbf{ASB} ($n{=}2000$)} \\
\cmidrule(lr){2-4}\cmidrule(lr){5-7}\cmidrule(lr){8-9}
System & ASR$\downarrow$ & Attack.\ Utility$\uparrow$ & Benign $\uparrow$ & \shortstack{Any\\leak$\downarrow$} & C2$\downarrow$ & C5$\downarrow$ & ASR$\downarrow$ & Util$\uparrow$ \\
\midrule
Plain       & 31.48 & 47.22 & 70.10 & 93.95 & 92.74 & 51.61 & 91.65 & 0.85 \\
IPIGuard    & 15.26 & 51.99 & 64.95 & 78.63 & 50.00 & 57.06 & 99.60 & 22.05 \\
DRIFT       &  2.38 & 52.78 & 59.79 & 72.18 & 47.98 & 53.63 & 87.75 & 0.45 \\
Faramesh    & 22.89 & 47.69 & 60.82 & 93.95 & 92.74 & 51.61 & \textbf{0.00} & 3.70 \\
SecureClaw  & \textbf{0.64} & \textbf{56.60} & 60.82 & \textbf{3.23} & \textbf{0.20} & \textbf{0.00} & \textbf{0.00} & \textbf{88.90} \\
\bottomrule
\end{tabular}
}
\end{table}

Table~\ref{tab:cross-bench} shows that, among the evaluated common-harness baselines in our study, SecureClaw is the only same-harness configuration that jointly drives unauthorized commits to zero/near-zero on the effect benchmarks and reduces AgentLeak's attacked parity-lane leakage to 3.23\% (16/496). On effectful benchmarks it reaches 0\% ASR on ASB and 0.64\% ASR on AgentDojo (4/629). On AgentLeak, the residual is C1-dominant rather than an internal-relay failure: 15/496 cases remain on C1, one isolated case appears on C2, and C5 stays at 0/496. Faramesh matches the effect boundary on ASB but leaves AgentLeak's internal channels open because plaintext still lives inside the runtime. IPIGuard and DRIFT reduce some attack success rates, but their enforcement root remains runtime-local. 
\paragraph{Utility under attack.}
SecureClaw retains 88.90\% utility on ASB at 0\% ASR. A separate recovery ablation (Appendix~\ref{app:recovery-ablation}) shows that safe continuation raises matched-sample ASB utility from 70\% to 86\% without increasing ASR. On the AgentLeak parity run, attack utility is $74.6\%$ and benign utility is $77.0\%$ (Table~\ref{tab:agentleak-utility}), so the low parity-lane leak rate is not achieved by collapsing the workflow into refusal.

\subsection{Mechanism isolation: why both pillars are necessary}
\label{sec:exp:ablation}

\begin{table}[t]
\centering
\caption{Mechanism ablation on a matched AgentLeak slice (24 attacked + 24 benign valid rows per configuration; C2/C5 leak rates are reported on attacked rows only). Boundary = executor-side PREVIEW$\rightarrow$COMMIT; Handles = reference-ID confinement on the read path.}
\label{tab:ablation}
\small
\begin{tabular}{lcccc}
\toprule
Config & Boundary & Handles & C2 Leak & C5 Leak \\
\midrule
Full        & \checkmark & \checkmark & \textbf{0.0\%} & \textbf{0.0\%} \\
$-$Boundary & $\times$   & \checkmark & \textbf{0.0\%} & \textbf{0.0\%} \\
$-$Handles  & \checkmark & $\times$   & 20.8\% & 8.3\% \\
$-$Both     & $\times$   & $\times$   & 25.0\% & 4.2\% \\
\bottomrule
\end{tabular}
\end{table}
Table~\ref{tab:ablation} isolates mechanisms on a matched 48-row AgentLeak slice. Removing handles (reference-ID confinement) raises attacked-row C2 leakage from $0/24$ to $5/24$ or $6/24$, and C5 leakage reappears ($2/24$ or $1/24$). Removing the execution boundary while keeping handles leaves both C2 and C5 at $0/24$. The confidentiality gain therefore comes from handle confinement, not from the write boundary accidentally masking it.

The converse appears when APIs that execute external actions are stressed directly: handles do not stop a compromised runtime from attempting unauthorized commits. That role belongs to the executor-side boundary and is validated in the bypass suite. Together, the ablation and adversarial validation support the main system's claim: the two pillars protect different surfaces, and neither subsumes the other.

\paragraph{Deny-aware recovery is the source of high blocked-attack utility.}
We measure recovery separately to show that SecureClaw's usability does not come from weaker enforcement. On a matched ASB sample ($n{=}50$), deny-aware recovery raises utility from 70\% to 86\% while preserving 0\% ASR. On a matched AgentLeak sweep ($n{=}30$), it raises utility from 40.0\% to 66.7\%, again with no ASR cost, while adding about 0.33 extra tool turns per task. SecureClaw stays usable because it offers safe next steps after denial, not because it relaxes the commit invariant.

\subsection{Robustness under runtime compromise}
\label{sec:exp:bypass}

Behavioral prompt-injection benchmarks are necessary but not sufficient because the formal claims also cover a fully adversarial runtime. We therefore add direct adversarial stress tests against the commit path.

Across 29 bypass scenarios spanning authorization-artifact tampering, malformed authorization material, request-binding violations, session/caller mismatches, context-binding attacks, temporal replay, and chained attacks, SecureClaw denies all 25 adversarial cases and accepts all 4 positive controls. In a payment case study, the executor blocks 4/4 unauthorized commits while allowing the legitimate transfer. These results show that the commit path respects the request-bound invariant even when the runtime constructs adversarial inputs.

\subsection{Residual failures and overhead}
\label{sec:exp:residual}

\begin{table}[t]
\centering
\caption{Residual failures after SecureClaw closes the two target surfaces.}
\label{tab:residual-main}
\small
\begin{tabular}{p{0.27\textwidth}p{0.12\textwidth}p{0.46\textwidth}}
\toprule
Residual class & Count & Interpretation \\
\midrule
Wrong principal or missing context binding & $3/629$ & The action family is permitted, but the target principal is attacker-chosen; richer principal and object binding is needed. \\
Broader-scope workspace action & $1/629$ & A workspace mutation remains individually allowed yet wrong for the task; coarse allow bits are insufficient. \\
Residual information leakage & $16/496$ & Fifteen cases remain on C1 and one on C2; none appear on C5. The residual is concentrated in final outputs rather than systematic internal relay.\\
\bottomrule
\end{tabular}
\end{table}

The remaining failures show what is left once the two target surfaces are controlled. The dominant residual is no longer unauthorized execution or internal-channel plaintext relay, but behavior that the policy permits but that does not match the intended task. On AgentDojo, all 4/629 remaining attack successes stay within the allowed action space: three involve a wrong principal or missing context binding, and one applies an allowed workspace action to the wrong object. On AgentLeak, 16/496 attacked parity-lane scenarios still leak: fifteen are C1 final-output cases and one is an isolated C2 hit. We observe no C5 leakage and no systematic reopening of internal relay, shifting the residual problem to policy-allowed misalignment and user-visible information release.

The remaining failures concern which actions and information releases the policy permits; they do not bypass the core invariants. The residual leakage motivates evaluating more selective content release through the planning summary. Richer principal, object, and phase binding would address the remaining AgentDojo failure classes.

Our local policy-check diagnostic takes about 15\,ms with cleartext policy lookup and 149\,ms with the distributed two-evaluator configuration, including gateway orchestration (Appendix~\ref{app:latency-details}). Relative to LLM inference latencies of 500--2000\,ms, these measured costs are not dominant.

\paragraph{Model dependence and evaluation protocol.}
Full-inventory transfer runs preserve $0/2000$ ASR on each of five ASB model snapshots, with utility from 83.6\% to 94.2\% (Appendix~\ref{app:full-model-transfer}). The original AgentLeak spot checks remain in Appendix~\ref{app:cross-model}; the primary $16/496$ result is unchanged. Appendix~\ref{app:accommodation-study} quantifies reference resolution and confirmation handling, and Appendix~\ref{app:field-identification} measures the effect of missed protected fields. These checks distinguish trusted enforcement from model-dependent planning and output behavior.

\section{Additional related work}
\label{sec:related}

\textbf{Benchmarks and runtime defenses.}
AgentDojo, ASB, AgentLeak, WASP, and tool-selection attacks have shifted evaluation from prompt classification toward dynamic agents that read untrusted content, call tools, and maintain state~\citep{debenedetti2024agentdojo,zhang2025agent,agentleak2026,wasp2025,toolselect2025}. Runtime-centric defenses such as IPIGuard and DRIFT make the planner more robust by isolating injected instructions or constraining tool dependencies~\citep{ipiguard2025,drift2025}. SecureClaw uses such defenses as common-harness baselines when their public surface is comparable, but studies a different claim: an injected or compromised runtime should not be the reference monitor for external effects or sensitive plaintext.

\textbf{Execution control and agent authorization.}
Faramesh, Progent, PCAS, and AC4A move enforcement toward policy or capability checks at the agent action boundary~\citep{faramesh2026,progent2025,pcas2026,ac4a2026}. SecureClaw combines established authorization and data-protection mechanisms under a joint contract for an adversarial runtime. Its contribution is the enforced relationship between the two controls and the evidence that each is necessary. On the write path, PREVIEW$\rightarrow$COMMIT binds the policy decision to the exact canonical request later executed, closing the time-of-check/time-of-use gap between planner inspection and the external API call. Authorization at the action boundary alone does not remove secrets already resident in planner state: a blocked commit cannot undo leakage through inter-agent messages or memory. This motivates coupling effect mediation with a separate trusted read gateway.

Faramesh overlaps with canonical action authorization, while CaMeL covers both control and data flow through a capability-enforcing interpreter~\citep{camel2025}. SecureClaw preserves the native tool-calling loop and places both checks in trusted components outside that loop. ClawGuard enforces user-confirmed, task-scoped rules at every tool-call boundary within a single tool-call pipeline~\citep{clawguard2026}; SecureClaw instead re-verifies a request-bound approval in a separate executor and keeps raw protected values in trusted storage. Concurrent work on Twin Agent also separates privileges, but between two LLM agents: an Explore Agent reads untrusted content and passes compact hints to a privileged Safe Agent~\citep{twinagent2026}. SecureClaw does not trust an LLM for security enforcement: the runtime is untrusted, the summary core follows fixed rules, and a non-LLM executor re-verifies each request. We view the two designs as complementary.

\textbf{Information flow, capabilities, and policy expressiveness.}
The read path draws on information-flow control, declassification, and capability systems~\citep{rtbas2025,costa2025securing,sabelfeld2005declassification,saltzer1975protection,watson2010capsicum,sabelfeld2003language}. Unlike settings where the protected computation is an enforceable program, LLM agents have opaque hidden state and natural-language channels. SecureClaw therefore makes raw protected values unrepresentable in the runtime except through reference IDs, and treats short planning summaries as explicit declassification rather than benign preprocessing. Reference IDs are not ambient authorities for the model: they are resolved only at trusted endpoints after caller, session, destination, and policy checks. Finally, practical authorization systems such as Zanzibar show that utility depends on expressive principal, object, and relation bindings~\citep{pang2019zanzibar}. SecureClaw's remaining benign-utility loss fits this pattern: richer bindings, phase-specific capabilities, and more precise declassification budgets can improve utility while preserving the request-binding and plaintext-residency invariants.

\section{Limitations and broader impact}
\label{sec:limitations}

SecureClaw reduces unsafe agentic effects and plaintext exposure, but it does not guarantee semantic correctness within policy-allowed actions and still incurs benign-utility loss on AgentDojo. The primary common-harness comparison uses a single backbone (GPT-4o-mini). The additional ASB snapshots in Appendix~\ref{app:full-model-transfer} retain zero observed attack success, but task utility varies with the backbone. Final-answer leakage depends on model output behavior, and for AgentLeak this paper reports only spot checks on other backbones (Appendix~\ref{app:cross-model}). The residuals motivate stronger principal, object, and phase binding in the policy layer~\citep{pang2019zanzibar}. Protected-field identification is recall-critical: controlled misses increase leakage, while conservative over-marking is inexpensive in the tested range (Appendix~\ref{app:field-identification}). 

The main deployment risk is false assurance: if protected fields are misclassified or APIs that execute external actions remain unmediated, the architecture's boundary guarantees no longer apply. SecureClaw therefore should be deployed with explicit inventories of the APIs that execute external actions, schema audits, and conservative confirmation policies for irreversible actions. Positively, the architecture can reduce unauthorized external effects and internal relay of sensitive values in high-stakes agent workflows without requiring the untrusted runtime to become a trusted security monitor. SecureClaw trusts the gateway, policy engine, and executor; Appendix~\ref{app:proofs:compromise} details compromise effects, and the optional distributed policy variant relies on non-collusion.

\section{Conclusion}
\label{sec:conclusion}

SecureClaw is built on a simple claim: securing LLM agents requires separating control of external effects from access to sensitive plaintext. The executor-side PREVIEW$\rightarrow$COMMIT boundary makes committed actions request-bound and non-bypassable, while the gateway-side reference-ID layer keeps internal channels symbolic rather than plaintext-bearing. Across AgentDojo, AgentLeak, and ASB, SecureClaw is the only evaluated same-harness system in our study that drives unauthorized effects to zero or near zero while remaining usable under attack through deny-aware recovery. The ablation and bypass results show that the gains are structural: the two mechanisms protect different failure modes, and high-stakes deployments need both.

\begin{ack}
We thank the anonymous reviewers for their helpful comments and suggestions.
Research in part supported by the German Research Foundation (DFG), SPP 2378 - ReNO-2, grant 511099228, 2025--2029.
Yuhan Ma is employed by Huawei Technologies Düsseldorf GmbH. Huawei provided no funding or other support for this work.
The authors declare no other competing financial interests.
\end{ack}

\clearpage

\bibliographystyle{plainnat}
\bibliography{references}

\appendix
\section{Extended experimental evidence}
\label{app:extended-results}

This appendix substantiates the main empirical claims with exact counts, uncertainty bounds, benchmark breakdowns, and robustness checks. The evidence is consistent across the different evaluation blocks: the execution boundary closes the effect surface, reference IDs (opaque handles) close the internal relay surface, and deny-aware recovery recovers substantial under-attack utility without reopening either one.

\paragraph{Statistical reporting.}
All reported rates are Bernoulli proportions.
For the main comparisons we report exact counts, 95\% confidence intervals, and two-sided Fisher exact tests where pairwise significance is informative. For zero-event rows, the appendix uses exact Clopper--Pearson upper bounds rather than asymptotic normal approximations.

\subsection{Exact counts and selected uncertainty}
\label{app:uncertainty}

Table~\ref{tab:uncertainty-main} summarizes the main empirical results with exact counts and confidence intervals.
It complements the percentage tables in the main paper.

\begin{table}[h]
\centering
\caption{Selected exact counts, 95\% exact confidence intervals, and exact tests.}
\label{tab:uncertainty-main}
\small
\setlength{\tabcolsep}{4pt}
\begin{tabular}{p{0.23\textwidth}p{0.26\textwidth}p{0.26\textwidth}p{0.15\textwidth}}
\toprule
Metric & SecureClaw & Comparator & Two-sided Fisher $p$ \\
\midrule
AgentDojo ASR & $4/629=0.64\%$ $[0.17,1.62]$ & DRIFT $15/629=2.38\%$ $[1.34,3.90]$ & $1.83 \times 10^{-2}$ \\
AgentLeak parity-lane any leak & $16/496=3.23\%$ $[1.85,5.19]$ & DRIFT $358/496=72.18\%$ $[68.01,76.08]$ & $1.27 \times 10^{-128}$ \\
ASB ASR & $0/2000=0.00\%$ $[0.00,0.18]$ & DRIFT $1755/2000 = 87.75\%$ $[86.23,89.16]$ & $<10^{-700}$ \\
ASB utility & $1778/2000=88.9\%$ $[87.44,90.24]$ & Faramesh $74/2000=3.7\%$ $[2.92,4.63]$ & $<10^{-700}$ \\
\bottomrule
\end{tabular}
\end{table}

Table~\ref{tab:uncertainty-main} makes clear that the core security results remain strong even under exact accounting.

\subsection{Evaluation methodology}
\label{app:methodology}

\paragraph{Model and inference settings.}
All primary comparisons use \texttt{gpt-4o-mini-2024-07-18} through OpenRouter with temperature 0.
Where the provider exposes a seed parameter, we request a fixed seed.
Because API-hosted inference is not guaranteed to be perfectly deterministic, repeated runs are interpreted as stability checks rather than as bitwise replicas.
SecureClaw is evaluated with the deployed bounded-summary interface throughout the benchmark results in the paper; the security appendix separately analyzes a handle-only reference case and measures the deployed summary operator's confidentiality behavior.

\paragraph{Benchmark protocols.}
\textbf{AgentDojo} uses v1.1.2 with the \texttt{important\_instructions} attack family, yielding 629 attacked rows and 97 benign rows across banking, slack, travel, and workspace.
\textbf{AgentLeak} uses its official split of 996 total scenarios: 500 benign and 496 attacked scenarios for the parity lane covering $C1/C2/C5$, plus channel-specific lanes for $C3$, $C4$, and $C6$.
The main-paper ``Any leak'' number refers to the 496 attacked parity-lane scenarios.
\textbf{ASB} is evaluated on five direct prompt-injection styles (naive, \texttt{escape\_characters}, \texttt{fake\_completion}, \texttt{context\_ignoring}, and \texttt{combined\_attack}) with 400 rows each, for 2000 total rows per baseline.

  \paragraph{Baseline fidelity.}
  Table~\ref{tab:baseline-fidelity} summarizes how each same-harness baseline is instantiated.

\begin{table}[t]
  \centering
  \caption{Baseline provenance and common-harness instantiation.}
  \label{tab:baseline-fidelity}
  \footnotesize
  \setlength{\tabcolsep}{4pt}
  \renewcommand{\arraystretch}{1.08}
  \begin{tabularx}{\textwidth}{@{}lYYY@{}}
  \toprule
  System & Evaluation implementation & Modification scope & Common-harness entry point \\
  \midrule

  Plain
  & Benchmark no-defense configuration.
  & No defense code or wrapper.
  & Native benchmark runners and direct AgentLeak outputs. \\

  IPIGuard
  & Authors' upstream code, vendored for evaluation.
  & Harness-only patches for API compatibility, retries, resume, and benchmark integration; defense logic unchanged.
  & IPIGuard AgentDojo runner, AgentLeak pipeline adapter, and ASB-style construct/traverse loop. \\

  DRIFT
  & Authors' upstream code, vendored for evaluation.
  & Harness-only patches for hosted-model access, retries, logging, and benchmark integration; defense logic unchanged.
  & DRIFT pipeline, AgentLeak tool-execution loop, and DRIFT-provided ASB fork. \\

  Faramesh
  & Upstream Faramesh core, vendored for evaluation.
  & Core runtime unchanged; our harness supplies benchmark adapters and policies.
  & Local Faramesh server mediating AgentDojo, AgentLeak, and ASB-style tool calls. \\

  \bottomrule
  \end{tabularx}
\end{table}

\paragraph{Mediated-interface accommodations.}
SecureClaw uses two evaluation accommodations that arise from mediated interfaces rather than from relaxed security:
\textbf{alias resolution} and \textbf{auto-confirmation}.
Alias resolution maps reference IDs such as \texttt{EMAIL\_REF\_3} to task-relevant literals in trusted downstream tool calls and before exact-string utility scoring. The resolved recipient, channel, and context still undergo policy authorization and request binding.
Auto-confirmation emulates a user click in non-interactive benchmark runners that otherwise provide no confirmation channel. It is applied only to requests that are already policy-allowed and remains bound to $(a,p,h)$, so it cannot convert a policy-denied attack proposal into an allowed commit. Appendix~\ref{app:accommodation-study} measures where each accommodation is actually used and separates AgentLeak release confirmation from the other benchmarks.

\paragraph{Run-to-run variation.}
Hosted inference introduces provider-level nondeterminism, so we use two complementary stability checks.
First, three completed banking/slack AgentDojo replications keep ASR at 0\% on those completed suites while utility varies by only single-digit points; Table~\ref{tab:replication} gives the details.
Second, a matched three-seed AgentLeak-style SecureClaw sweep ($n=30$ per seed) keeps ASR at 0.0\% across all seeds with utility mean 64.4\% and standard deviation 8.4, and an independent ASB rerun at $n=50$ also preserves 0\% ASR with 86.0\% utility.
The qualitative conclusion is stable: security is much less sensitive than utility to hosted-model nondeterminism.

\begin{table}[h]
\centering
\caption{Replication stability on AgentDojo banking + slack attack suites.}
\label{tab:replication}
\small
\begin{tabular}{lcccc}
\toprule
 & \multicolumn{2}{c}{Banking ($n=144$)} & \multicolumn{2}{c}{Slack ($n=105$)} \\
\cmidrule(lr){2-3}\cmidrule(lr){4-5}
Run & ASR & Utility & ASR & Utility \\
\midrule
Run 1 (primary) & 0.0\% & 56.3\% & 0.0\% & 57.1\% \\
Run 2 & 0.0\% & 52.1\% & 0.0\% & 49.5\% \\
Run 3 & 0.0\% & 51.4\% & 0.0\% & 48.6\% \\
\midrule
Sample std.\ dev. & 0 pp & $\pm 2.7$ pp & 0 pp & $\pm 4.6$ pp \\
\bottomrule
\end{tabular}
\end{table}

\paragraph{Code availability}
\label{app:code-availability}
The SecureClaw implementation and evaluation code are available at \url{https://github.com/ponyyuhan/SecureClaw}. An archived version is available at \url{https://doi.org/10.5281/zenodo.23258171}. The repository includes the core runtime, benchmark adapters, configurations, evaluation and scoring scripts, setup instructions, and local tests. Our code is released under the MIT license; third-party components retain their original licenses. The documentation describes how to obtain benchmark assets and run the experiments. Generated results, raw model traces, and temporary development files are not bundled.
  
\paragraph{Compute resources.}
All experiments were run on a local macOS workstation (Apple M4 CPU, 16 GB RAM) with API-hosted LLM inference.
The gateway, handle store, executor, and policy-evaluation components execute locally; the reported latency numbers therefore reflect algorithmic and local orchestration overhead, not geographically distributed network round trips between separate servers.

\paragraph{Asset provenance and licenses.}
Table~\ref{tab:assets} lists the main third-party benchmark assets used in our experiments and their stated upstream licenses. Our same-harness baselines use upstream code or benchmark adapters as summarized in Table~\ref{tab:baseline-fidelity}; retained third-party components keep their original LICENSE files and attribution metadata.

\begin{table}[h]
\centering
\caption{Main third-party benchmark assets used in this paper.}
\label{tab:assets}
\small
\begin{tabular}{p{0.20\textwidth}p{0.27\textwidth}p{0.12\textwidth}p{0.28\textwidth}}
\toprule
Asset & Citation / upstream source & License & Use in this paper \\
\midrule
AgentDojo & \citep{debenedetti2024agentdojo} / official benchmark repository & MIT & Prompt-injection evaluation \\
AgentLeak & \citep{agentleak2026} / official benchmark repository & MIT & Internal-leakage evaluation \\
ASB & \citep{zhang2025agent} / official benchmark repository & MIT & Effect-side attack evaluation \\
\bottomrule
\end{tabular}
\end{table}

\subsection{Cross-model spot checks}
\label{app:cross-model}

The main paper keeps all full-scale baseline comparisons on a single backbone to preserve a common harness.
To test whether the security pattern is tied to that model, we also ran small SecureClaw-only spot replications on additional model families.
We use these as transfer checks rather than as full cross-model baseline rerankings.

\begin{table}[h]
\centering
\caption{SecureClaw cross-model transfer checks.}
\label{tab:cross-model}
\small
\begin{tabular}{p{0.16\textwidth}p{0.23\textwidth}p{0.11\textwidth}p{0.22\textwidth}p{0.12\textwidth}}
\toprule
Benchmark & Model & Rows & Security result & Utility \\
\midrule
ASB & Claude Sonnet 4 & 50 & ASR $=0/50=0.0\%$ & 94.0\% \\
ASB & Claude 3.5 Haiku & 50 & ASR $=0/50=0.0\%$ & 86.0\% \\
AgentLeak slice & Claude Sonnet 4 & 50 & $C1/C2/C5 = 0/0/0$ & 100.0\% \\
\bottomrule
\end{tabular}
\end{table}

These transfer checks preserve the qualitative security pattern off the primary backbone: SecureClaw keeps ASR at zero on both non-primary ASB runs and remains clean on the Claude Sonnet 4 AgentLeak slice.

\subsection{AgentLeak overview and per-channel leak rates}
\label{app:agentleak-channel}

The main paper uses the attacked parity lane's overall scenario-OR leak because it provides the closest single-number analogue to AgentDojo ASR and ASB ASR. Table~\ref{tab:agentleak-perchannel} gives the parity-lane per-channel breakdown and Table~\ref{tab:agentleak-channellane} reports the remaining channel lanes $C3/C4/C6$. SecureClaw's residual is C1-dominant: 15/496 parity-lane cases remain on C1, one isolated case appears on C2, and C5 stays at 0/496. 

\begin{table}[h]
\centering
\caption{AgentLeak parity-lane per-channel leak rates (\%) $\downarrow$ ($n=496$ attacked rows).}
\label{tab:agentleak-perchannel}
\small
\begin{tabular}{lccc}
\toprule
System & C1 Leak & C2 Leak & C5 Leak \\
\midrule
Plain & 33.47 & 92.74 & 51.61 \\
IPIGuard & 56.25 & 50.00 & 57.06 \\
DRIFT & 51.21 & 47.98 & 53.63 \\
Faramesh & 33.47 & 92.74 & 51.61 \\
SecureClaw & \textbf{3.02} & \textbf{0.20} & \textbf{0.00} \\
\bottomrule
\end{tabular}
\end{table}

\begin{table}[h]
\centering
\caption{AgentLeak channel-lane results: attack leak rate (\%) $\downarrow$ / benign allow rate (\%) $\uparrow$.}
\label{tab:agentleak-channellane}
\small
\begin{tabular}{l cc cc cc}
\toprule
& \multicolumn{2}{c}{C3} & \multicolumn{2}{c}{C4} & \multicolumn{2}{c}{C6} \\
\cmidrule(lr){2-3}\cmidrule(lr){4-5}\cmidrule(lr){6-7}
System & Leak$\downarrow$ & Allow$\uparrow$ & Leak$\downarrow$ & Allow$\uparrow$ & Leak$\downarrow$ & Allow$\uparrow$ \\
\midrule
Plain & 100.0 & 100.0 & 100.0 & 100.0 & 100.0 & 100.0 \\
IPIGuard & 24.03 & 100.0 & 33.61 & 100.0 & 0.00 & 100.0 \\
DRIFT & 6.98 & 99.21 & 14.75 & 99.80 & 0.00 & 98.90 \\
Faramesh & 100.0 & 100.0 & 100.0 & 100.0 & 100.0 & 100.0 \\
SecureClaw & \textbf{0.00} & \textbf{100.0} & \textbf{0.00} & \textbf{100.0} & \textbf{0.00} & 99.90 \\
\bottomrule
\end{tabular}
\end{table}

For completeness, Table~\ref{tab:agentleak-utility} reports the official strict-evaluator utility counts underlying the main-text AgentLeak utility sentence.

\begin{table}[h]
\centering
\caption{AgentLeak parity-lane utility counts for SecureClaw.}
\label{tab:agentleak-utility}
\small
\begin{tabular}{lcc}
\toprule
Slice & Successes & Rate \\
\midrule
Attacked rows & $370/496$ & 74.6\% \\
Benign rows & $385/500$ & 77.0\% \\
Overall & $755/996$ & 75.8\% \\
\bottomrule
\end{tabular}
\end{table}

These tables make the paper's security decomposition concrete.
Faramesh matches SecureClaw on ASB-style effect blocking but remains indistinguishable from Plain on $C2$ and $C5$ because it does not remove plaintext from the runtime.
SecureClaw reduces $C2$ leakage to $1/496$ and eliminates measured leakage on $C3$, $C4$, $C5$, and $C6$, while leaving $C1$ as the dominant residual channel.
That is exactly the behavior predicted by the architecture.

\subsection{Residual taxonomy}
\label{app:residual-taxonomy}

Table~\ref{tab:residual-taxonomy} groups the observed residuals by security meaning.

\begin{table}[h]
\centering
\caption{Manual residual taxonomy from the full runs and targeted cross-checks.}
\label{tab:residual-taxonomy}
\small
\begin{tabular}{p{0.20\textwidth}p{0.12\textwidth}p{0.22\textwidth}p{0.31\textwidth}}
\toprule
Evidence source & Residual count & Manual class & Interpretation \\
\midrule
AgentDojo full attack run & 3/629 & wrong principal or missing context binding & The action family is allowed, but the target principal is attacker-chosen; richer principal and object binding is needed. \\
AgentDojo full attack run & 1/629 & broader-scope workspace action & A workspace mutation remains individually allowed yet wrong for the task; coarse allow bits are insufficient. \\
AgentLeak attacked parity lane & 16/496 residuals & authorized-output-dominant residual & Fifteen cases remain on the final authorized output channel and one on C2; no C5 relay reappears. \\
\bottomrule
\end{tabular}
\end{table}

The residual profile is concentrated and interpretable. We do not claim that all $15/496$ surviving $C1$ cases share the same cause; the evidence here supports only the narrower conclusion that the residual is concentrated on the final authorized-output plane rather than on reopened internal relay.

\subsection{Ablation extended results}
\label{app:ablation-extended}

We report here the same matched slice used in the main paper after filtering rows with valid channel instrumentation: 48 valid rows per configuration (24 attack and 24 benign).
This slice is mechanism-focused and is not intended to reproduce the full AgentLeak Plain distribution; in particular, the $-\!$Both row should be read as a matched-slice internal control rather than as the Table~\ref{tab:cross-bench} Plain baseline.
The point of the ablation is not to optimize utility.
It is to identify the active mechanism behind the main-paper leakage reductions using a single denominator convention throughout.

\begin{table}[h]
\centering
\caption{Extended ablation on the matched AgentLeak slice. C2/C5 are reported on the 24 attacked rows; benign leak is reported on the 24 benign rows.}
\label{tab:ablation-extended}
\small
\resizebox{\textwidth}{!}{%
\begin{tabular}{lccccc}
\toprule
Config & Boundary & Handles & C2 Leak & C5 Leak & Benign Leak \\
\midrule
Full & \checkmark & \checkmark & $0/24$ (0.0\%) & $0/24$ (0.0\%) & $0/24$ (0.0\%) \\
$-$Boundary & $\times$ & \checkmark & $0/24$ (0.0\%) & $0/24$ (0.0\%) & $0/24$ (0.0\%) \\
$-$Handles & \checkmark & $\times$ & $5/24$ (20.8\%) & $2/24$ (8.3\%) & $4/24$ (16.7\%) \\
$-$Both & $\times$ & $\times$ & $6/24$ (25.0\%) & $1/24$ (4.2\%) & $4/24$ (16.7\%) \\
\bottomrule
\end{tabular}
}
\end{table}

Turning off the boundary while keeping handles leaves $C2$ and $C5$ at zero, which shows that those leak reductions do not come from the execution boundary.
Turning off handles while keeping the boundary reopens $C2$ to $5/24$ attacked rows and $C5$ to $2/24$ attacked rows; removing both gives $6/24$ and $1/24$ respectively.
The converse point is validated by the bypass suite: handles alone do not secure effectful sinks.
On this small matched slice, removing handles produces five or six additional C2 leaks among 24 attacked rows; the C5 differences involve only one or two rows. These counts support the mechanism comparison, while the small sample limits the precision of channel-specific rates.

\subsection{Recovery ablation details}
\label{app:recovery-ablation}

\begin{table}[h]
\centering
\caption{Recovery ablation on two matched sweeps. All configurations retain 0\% ASR.}
\label{tab:recovery-ablation}
\small
\begin{tabular}{llccccc}
\toprule
Benchmark & Config & $n$ & ASR & Utility & Refusals & Extra turns/task \\
\midrule
ASB & recovery on & 50 & 0.0\% & 86.0\% & 8.0\% & --- \\
ASB & recovery off & 50 & 0.0\% & 70.0\% & 30.0\% & --- \\
\midrule
AgentLeak-style & recovery on & 30 & 0.0\% & 66.7\% & 66.7\% & 1.33 \\
AgentLeak-style & recovery off & 30 & 0.0\% & 40.0\% & 60.0\% & 1.00 \\
\bottomrule
\end{tabular}
\end{table}

The two sweeps are consistent.
Recovery lifts utility by 16.0 points on ASB and 26.7 points on the matched AgentLeak-style subset, with no ASR cost in either case.
On the instrumented subset it adds only about 0.33 extra tool turns per task, which is small relative to the LLM-driven workflows we target.

\paragraph{Per-attack breakdown.}
With recovery enabled, utility is 60\% on naive, 80\% on \texttt{escape\_characters}, 100\% on \texttt{fake\_completion}, 100\% on \texttt{context\_ignoring}, and 90\% on \texttt{combined\_attack}.
Without recovery, the largest drops occur on naive and \texttt{escape\_characters}, where the denial message helps the runtime restart from a cleaner plan.
The key point is straightforward: SecureClaw's high ASB utility comes from safe continuation after denial, not from weaker enforcement.

\subsection{Latency details}
\label{app:latency-details}

\begin{table}[h]
\centering
\caption{Latency of the local policy-check diagnostic, including gateway orchestration.}
\label{tab:latency}
\small
\begin{tabular}{lccc}
\toprule
Configuration & Avg (ms) & p50 (ms) & p95 (ms) \\
\midrule
Distributed policy checks & 148.6 & 159.7 & 166.6 \\
Cleartext policy lookup & 15.4 & 12.8 & 32.2 \\
Policy-bypass control & 0.22 & 0.09 & 0.55 \\
\bottomrule
\end{tabular}
\end{table}

The diagnostic primarily issues policy-check requests rather than committed external actions. Its policy-bypass control disables policy evaluation and the external executor, so that row does not isolate executor verification cost. The distributed configuration is slower than cleartext lookup, while both remain below the 500--2000\,ms LLM inference range used for comparison. These measurements cover local orchestration, not network round trips between geographically separated services.

\subsection{Commit-path bypass suite details}
\label{app:bypass-details}

The bypass suite validates the implementation against the compromised-runtime threat model.
It intentionally goes beyond prompt injection and directly mutates authorization artifacts, timestamps, request bindings, and replay state. The suite includes both core single-service cases and a small number of optional distributed-variant checks.

\begin{table}[h]
\centering
\caption{Commit-path bypass suite, including optional distributed-variant cases: 29/29 expected outcomes (25 attacks denied, 4 controls allowed).}
\label{tab:e1-details}
\small
\setlength{\tabcolsep}{4pt}
\resizebox{\textwidth}{!}{
\begin{tabular}{llcc}
\toprule
Category & Bypass attempt & Expected & Observed \\
\midrule
\multirow{4}{*}{Authorization-share integrity}
  & Bit-flipped MAC in authorization share & DENY & DENY \\
  & Wrong share version ($v=2$) & DENY & DENY \\
  & Server-ID swap ($\mathit{sid}_0=1$) & DENY & DENY \\
  & Cross-commit mixup (mismatched authorization shares) & DENY & DENY \\
\midrule
\multirow{3}{*}{Missing shares}
  & Missing second authorization share & DENY & DENY \\
  & Empty commit (no authorization shares) & DENY & DENY \\
  & Null policy entries & DENY & DENY \\
\midrule
\multirow{5}{*}{Binding violations}
  & Request hash mismatch (text tamper) & DENY & DENY \\
  & Recipient tampered after commit & DENY & DENY \\
  & Channel tampered after commit & DENY & DENY \\
  & Domain tampered after commit & DENY & DENY \\
  & Action ID mismatch & DENY & DENY \\
\midrule
\multirow{3}{*}{Session/caller}
  & Session ID mismatch & DENY & DENY \\
  & Caller ID mismatch & DENY & DENY \\
  & Cross-session replay & DENY & DENY \\
\midrule
\multirow{3}{*}{Context binding}
  & External principal injection & DENY & DENY \\
  & Delegation-token identifier (JTI) injection & DENY & DENY \\
  & Context-hash injection & DENY & DENY \\
\midrule
\multirow{3}{*}{Temporal}
  & Expired freshness window (past) & DENY & DENY \\
  & Future-timestamped authorization & DENY & DENY \\
  & Replay of consumed action\_id & DENY & DENY \\
\midrule
  & Missing action\_id (empty) & DENY & DENY \\
\midrule
\multirow{2}{*}{Chained attacks}
  & Expired authorization + text tamper & DENY & DENY \\
  & MAC flip + recipient change & DENY & DENY \\
\midrule
  & Confirm-gated: no user\_confirm & DENY & DENY \\
\midrule
\multirow{4}{*}{Positive controls}
  & Valid dual authorization (primary) & ALLOW & ALLOW \\
  & Valid second request (different recipient) & ALLOW & ALLOW \\
  & Valid third request (different text) & ALLOW & ALLOW \\
  & Confirm-gated: with user\_confirm & ALLOW & ALLOW \\
\bottomrule
\end{tabular}
}
\end{table}

This table is the empirical companion to Theorem~\ref{thm:nbe}: every mutated path that would have undermined request-bound authorization is denied, and every legitimate control path still succeeds.

\subsection{Payment case-study details}
\label{app:payment-details}

\begin{table}[h]
\centering
\caption{Payment case study: the boundary blocks 4/4 unauthorized commits.}
\label{tab:payment-case}
\small
\begin{tabular}{lccc}
\toprule
Scenario & Attack vector & Boundary on & Boundary off \\
\midrule
Legitimate payment & Valid authorization & ALLOW & ALLOW \\
Amount manipulation & \$100 $\rightarrow$ \$10{,}000 & \textbf{DENY} & ALLOW \\
Recipient redirect & $\rightarrow$ attacker@evil.com & \textbf{DENY} & ALLOW \\
Forged authorization & Bit-flip MAC & \textbf{DENY} & ALLOW \\
Replay attack & Replay prior commit & \textbf{DENY} & ALLOW \\
\bottomrule
\end{tabular}
\end{table}

The payment case study is not the paper's only scenario, but it is a clean concrete illustration of the effect surface.
Once commit authority moves to the executor, post-preview mutation of amount, recipient, or authorization state no longer changes the committed action.

\section{Formal security guarantees}
\label{app:proofs}

This appendix makes the paper's formal contract explicit.
It states the guarantees SecureClaw is designed to enforce, isolates the assumptions needed for each one, and calibrates the deployed summary interface used in experiments.
Throughout, the runtime is adversarial unless stated otherwise.
The purpose is precision rather than breadth: the proofs justify the architecture under the stated trust boundary, while the experiments determine which component is responsible for the measured benchmark gains.

\subsection{System model, notation, and interfaces}
\label{app:proofs:primitives}

Let $\lambda$ denote the implicit security parameter; concrete bounds below use the handle length $\kappa$ and explicit oracle counts.
SecureClaw has five logical roles: an untrusted runtime $\mathcal{R}$, a trusted gateway $\mathcal{G}$, a trusted handle store $\mathcal{S}$ (the trusted data store of the main text), a trusted policy engine $\mathcal{P}$, and a trusted executor $\mathcal{X}$. Appendix~\ref{app:proofs:policy} analyzes an optional distributed realization in which the policy engine is instantiated as two evaluators $\mathcal{P}_0,\mathcal{P}_1$.
The threat model allows full corruption of $\mathcal{R}$.
The gateway, handle store, policy engine, and executor are trusted for the core authorization and confinement guarantees; Section~\ref{app:proofs:compromise} makes explicit which guarantees fail if any of those components is compromised. For the per-evaluator transcript-privacy result stated later, we refine the policy assumption: in the distributed realization, each evaluator follows the protocol and the pair does not collude.

\paragraph{Canonical request and binding digest.}
An effectful request is the tuple
\[
\rho = (\mathsf{intent}, \mathsf{caller}, \mathsf{session}, \mathsf{inputs}, \mathsf{ctx}),
\]
where $\mathsf{ctx}$ must contain every sink-relevant field needed to determine the realized external effect (a \emph{sink} is an interface at which an action takes external effect). Representative fields include recipient or attendee sets, channels, object identifiers, confirmation class, delegation identifiers, provenance hashes, policy/schema versions, and version identifiers for mutable referenced handles or objects.

The gateway and executor serialize requests using the same deterministic canonicalization procedure $\mathsf{Canon}(\cdot)$.
The binding digest is
\[
  h(\rho) := \mathrm{Bind}_{k_{\mathrm{bind}}}(\rho)
  := \mathrm{HMAC\text{-}SHA256}(k_{\mathrm{bind}}, \mathsf{Canon}(\rho)).
\]
The security role of $h(\rho)$ is narrow and important: it turns the executor's question from ``is this action type generally allowed?'' into ``is this exact canonical request the one that was authorized?''

\paragraph{Preview and commit interfaces.}
On preview, the gateway canonicalizes $\rho$, queries policy, and returns an action identifier $a$, a program identifier $p$, and an authorization artifact
\[
  \pi = (\mathsf{meta}, t), \qquad
 t = \mathsf{Tag}_{K_{\mathrm{auth}}}(\mathsf{Canon}(\mathsf{meta})),
 \]
where $\mathsf{meta}$ binds $(a,p,h,\tau,y)$, $\tau$ is the freshness timestamp/window material, and $y$ is the policy decision. The authentication key $K_{\mathrm{auth}}$ is distinct from the request-binding key $k_{\mathrm{bind}}$; SecureClaw assumes key separation and domain separation between artifact authentication and request binding.
On commit, the executor recomputes $h(\rho)$ and accepts only if all of the following hold:
\begin{enumerate}
    \item the authorization artifact carries a valid MAC under the policy-service authentication key;
    \item the artifact binds the tuple $(a,p,h)$;
    \item the recomputed digest equals the digest bound inside the artifact;
    \item timestamps fall inside the freshness window;
    \item the action identifier has not already been consumed by the replay guard;
    \item the reconstructed policy result is \texttt{allow};
    \item if the action requires user confirmation, a confirmation token bound to $(a,p,h)$ is present; and
    \item any required handle-resolution or sink-specific validation succeeds before the effectful sink is reached.
\end{enumerate}
Any failure is fail-closed. Appendix~\ref{app:proofs:policy} gives the optional distributed specialization that replaces $\pi$ with two separately authenticated shares for the same digest.

\paragraph{Handles, dereference, and read interfaces.}
Sensitive values identified by the deployment schema or policy are stored in a trusted handle table. Under-classified sensitive fields are outside the handle-confinement claim.
A handle record contains a fresh handle identifier $\mathsf{hid}$ together with caller, session, object version, time-to-live, and sink constraints.
The runtime may carry, store, or transmit $\mathsf{hid}$ as a symbol, but only trusted components may dereference it.
The deployed read interface returns $\mathsf{hid}$ together with an explicit bounded declassification $D(v)$.
For proof structure we also refer to a handle-only reference case in which the runtime receives only $\mathsf{hid}$ and non-sensitive metadata.
All benchmark results in the paper use the bounded-summary interface because current public benchmarks require natural-language read access for many tasks.

\paragraph{Effectful sinks and explicit declassification.}
An effectful sink is any interface that can commit an irreversible or externally visible side effect.
An explicit declassification event is any interface that intentionally releases information from a protected value to the runtime, including the read-summary interface and tightly scoped recovery messages.
Everything else is treated as non-declassifying by default.

\subsection{Assumption partition}
\label{app:proofs:assumptions}

Table~\ref{tab:assumption-partition} separates assumptions by the property they support.
This partition matters because the execution boundary, the handle layer, the summary interface, and the optional distributed policy-engine refinement rely on different technical premises.

\begin{table}[h]
\centering
\caption{Assumption partition by security property.}
\label{tab:assumption-partition}
\small
\setlength{\tabcolsep}{5pt}
\renewcommand{\arraystretch}{1.18}
\begin{tabularx}{\linewidth}{>{\bfseries\RaggedRight\arraybackslash}p{0.18\linewidth}YY}
\toprule
Property & Security objective & Assumptions \\
\midrule
Execution boundary
& No unauthorized side effects.
& MAC EUF-CMA; deterministic canonicalization; bounded request-binding collision term $\varepsilon_{\mathsf{bind}}(q_{\mathsf{bind}})$ over all digest evaluations in the game; confirmation binding when required; bounded clock drift; atomic replay check-and-mark; complete mediation of effectful sinks. \\
\midrule
Handle confinement
& No plaintext reaches the runtime except through declared interfaces.
& Complete mediation of sensitive reads; correct classification of protected fields according to the deployment schema/policy; handle identifiers are sampled uniformly from $\{0,1\}^{\kappa}$ with $\kappa \ge 128$ or generated pseudorandomly from a fresh nonce under a secret handle key with PRF loss $\varepsilon_{\mathsf{hid}}$; unsuccessful dereference attempts reveal only a generic denial bit; at most $N_{\mathsf{live}}$ valid handles are outstanding at a time; dereference is allowed only at trusted endpoints. \\
\midrule
Bounded-summary interface
& Quantified exposure from the deployed read interface.
& Fixed summary operator $D$ for the evaluated deployment; any optional post-processor depends only on $D(v)$ and independent randomness; leakage to the runtime occurs only through metadata, handles, and authorized summaries. \\
\midrule
Distributed policy evaluation
& Single-evaluator transcript privacy.
& Two-server private information retrieval (PIR) security against one server; honest-but-curious secure multi-party computation (MPC) privacy; passive non-colluding servers; fixed-shape routing; transcript leakage is limited to $L_{\mathsf{policy}}$. \\
\bottomrule
\end{tabularx}

\end{table}

\paragraph{Assumptions not made.}
The formal results do not assume an honest runtime, successful prompt filtering, or semantically perfect summaries.
The runtime is adversarial throughout.
Summary exposure is modeled explicitly through the operator $D$ rather than assumed away.
The transcript-privacy result for distributed policy evaluation is deliberately scoped to the non-colluding, single-evaluator-view setting stated below.

\subsection{Request-bound authorization integrity}
\label{app:proofs:nbe}

\paragraph{Security game.}
The challenger samples the artifact-authentication key $K_{\mathrm{auth}}$ and request-binding key $k_{\mathrm{bind}}$ and exposes a preview oracle $\mathcal{O}_{\mathsf{pv}}(\rho)$ and a commit oracle $\mathcal{O}_{\mathsf{cm}}(\rho,\pi)$.
On preview input $\rho$, $\mathcal{O}_{\mathsf{pv}}$ returns a fresh action identifier $a$, a program identifier $p$, a policy decision $y$, and a valid authorization artifact $\pi$, while recording $(a,p,h(\rho),y,\rho)$ in a set $\mathcal{S}$.
On commit input $(\rho,\pi)$, $\mathcal{O}_{\mathsf{cm}}$ runs the executor verification logic, including freshness and atomic replay check-and-mark, returns only accept/deny, and records accepted action identifiers.
The adversary may adaptively query both oracles.
For an accepting commit query, write the adversary-supplied input as $(\rho^\star,\pi^\star)$ and the accepted metadata as $(a^\star,p^\star,h^\star)$.
The adversary wins if any accepting commit query satisfies at least one of the following:
\begin{enumerate}
    \item \textbf{No authorization:} there is no tuple $(a^\star,p^\star,h^\star,\texttt{allow},\rho') \in \mathcal{S}$ matching the accepted authorization metadata.
    \item \textbf{Binding violation:} there exists $(a^\star,p^\star,h^\star,\texttt{allow},\rho') \in \mathcal{S}$, but $\mathsf{Canon}(\rho^\star) \neq \mathsf{Canon}(\rho')$ or the two requests realize different sink effects while both are accepted under the same bound digest.
    \item \textbf{Replay:} an already consumed action identifier is accepted again.
\end{enumerate}

Let $q_{\mathsf{bind}}$ denote the total number of canonical requests whose binding digest is ever evaluated in the game, including all adversarial commit candidates.

\begin{theorem}[Request-bound authorization integrity]
\label{thm:nbe}
Assume that
(i) the MAC is EUF-CMA secure,
(ii) $\mathrm{Canon}(\cdot)$ is deterministic, identical at the gateway and executor, uniquely decodable for authenticated metadata, and semantically complete for mediated sink effects: for requests reaching the same mediated sink, equal canonical strings imply the same realized external effect,
(iii) $\mathsf{ctx}$ includes all sink-relevant fields that can change the realized external effect, including version identifiers for mutable referenced objects when applicable,
(iv) the request-binding digest has bounded collision term $\varepsilon_{\mathsf{bind}}(q_{\mathsf{bind}})$ over distinct canonical byte strings,
(v) any required confirmation token is bound to $(a,p,h)$,
(vi) freshness checks are enforced,
(vii) replay check-and-mark is atomic and crash-safe, and
(viii) all effectful sinks are mediated by the executor.
Then for any PPT adversary $\mathcal{A}$, there exists a MAC-forgery reduction $\mathcal{B}$ such that
\[
\Pr[\mathcal{A}\text{ wins}]
\le
\mathrm{Adv}^{\mathrm{euf\text{-}cma}}_{\mathsf{MAC}}(\mathcal{B})
+
\varepsilon_{\mathsf{bind}}(q_{\mathsf{bind}}).
\]
\end{theorem}

\begin{proof}
We use a short sequence of games.

\textbf{Game $G_0$.}
This is the real authorization game.

\textbf{Game $G_1$.}
Abort if the executor accepts an authorization artifact for an authenticated message $\mathsf{Canon}(\mathsf{meta})$ that was not previously signed by the preview oracle. The change from $G_0$ to $G_1$ is bounded by
\[
\mathrm{Adv}^{\mathrm{euf\text{-}cma}}_{\mathsf{MAC}}(\mathcal{B}).
\]

\textbf{Game $G_2$.}
In addition to the previous abort rule, abort if an accepting commit for $\rho$ uses a bound digest that was previously recorded for a previewed request $\rho'$ with $\mathrm{Canon}(\rho) \ne \mathrm{Canon}(\rho')$, or if two accepted commit requests have distinct canonical byte strings under the same bound digest. The change from $G_1$ to $G_2$ is at most $\varepsilon_{\mathsf{bind}}(q_{\mathsf{bind}})$. The case in which two requests have the same canonical string but different realized sink effects is ruled out by the semantic completeness premise for canonicalization, not by the digest-collision term.
 
Conditioned on no abort in \(G_2\), every accepting commit query corresponds to a previously previewed authenticated artifact for the same bound digest and an allow decision. Because all sink-relevant fields appear in \(\mathsf{Canon}(\rho)\), any post-preview mutation that changes the realized external effect induces a distinct canonical byte string unless it is absorbed by a MAC forgery, a binding collision, or an implementation bug in the stated canonicalization premise. Under the remaining operational premises, including freshness enforcement, replay-safe consumption, confirmation binding when required, and complete mediation of effectful sinks, the executor can therefore accept only the exact previously previewed request once. None of the adversary's winning conditions can occur. A union bound over the game hops gives the stated inequality.
\end{proof}

\paragraph{Instantiating the binding term.}
The proof keeps $\varepsilon_{\mathsf{bind}}(q_{\mathsf{bind}})$ explicit because the architecture only needs a request-binding digest with negligible collision probability over the canonical request space.
In the evaluated system, $\mathrm{Bind}$ is instantiated with HMAC-SHA256 over canonical bytes.
Under the usual PRF-style heuristic for HMAC, $\varepsilon_{\mathsf{bind}}(q_{\mathsf{bind}})$ behaves like a standard collision term in the digest length; the theorem does not require a stronger, non-standard assumption.

Theorem~\ref{thm:nbe} captures unauthorized execution relative to the bound request. It does not claim that every policy-allowed request is semantically aligned with the user's latent objective; wrong-principal or wrong-object actions inside the policy-allowed region therefore lie outside its scope. Freshness, replay safety, confirmation binding, and complete mediation are deployment premises rather than cryptographic loss terms. The bypass evaluation in Appendix~\ref{app:bypass-details} stress-tests these premises empirically.

\paragraph{Single-service scope and optional distributed variant.}
The theorem above states the core single-service guarantee used throughout the main paper. Appendix~\ref{app:proofs:policy} analyzes an optional distributed policy-evaluation refinement; it does not change the effect/plaintext split or the meaning of request binding.

\subsection{Runtime confidentiality under opaque handles}
\label{app:proofs:sm}

\paragraph{Security game.}
The challenger initializes the trusted handle store and gives the adversary full control of the runtime.
The adversary may request protected reads, move handles across internal channels, and attempt guessed dereferences through any interface that is supposed to reject unknown or unauthorized handles.
It wins if it distinguishes which of two challenger-chosen secret values is stored behind a protected handle with non-negligible advantage, without triggering an explicit declassification event.

\paragraph{Base leakage function.}
Let $L_{\mathsf{handles}}$ include handle identifiers issued to the runtime, non-sensitive type metadata, TTL values, policy decision bits, and any declared recovery templates.
It excludes all other functions of sensitive plaintext.

\begin{theorem}[Handle-only confinement]
\label{thm:sm}
Assume the handle-only reference case: the runtime receives only $\mathsf{hid}$ and non-sensitive metadata; handle identifiers are sampled independently and uniformly from $\{0,1\}^{\kappa}$ with $\kappa \ge 128$ or generated by a PRF from a fresh nonce under a secret handle key, unsuccessful dereference attempts reveal only a generic denial bit, and at most $N_{\mathsf{live}}$ valid handles are simultaneously outstanding in the runtime's view. Let $\varepsilon_{\mathsf{hid}}=0$ for uniform handles and let $\varepsilon_{\mathsf{hid}}$ be the PRF distinguishing loss for PRF-generated handles.

Then for any adversary making at most $q_h$ online handle-guessing attempts, the runtime's view is simulatable from $L_{\mathsf{handles}}$ up to additive advantage
\[
\min\!\left(1,\frac{q_h N_{\mathsf{live}}}{2^{\kappa}}\right)+\varepsilon_{\mathsf{hid}}.
\]
\end{theorem}

\begin{proof}
Construct a simulator that replaces every protected value with an independently sampled handle identifier and exposes only the metadata in $L_{\mathsf{handles}}$. If the implementation uses PRF-generated handles rather than direct uniform sampling, first replace those handles by uniform strings, incurring $\varepsilon_{\mathsf{hid}}$.

Because the handle-only reference case never releases plaintext directly to the runtime, the simulated view and the real view are identical unless the adversary successfully guesses a live handle identifier that was not already issued to it and uses that guess to trigger a privileged dereference path.

At any point there are at most $N_{\mathsf{live}}$ valid targets in the handle namespace.
A single uniform guess therefore succeeds with probability at most $N_{\mathsf{live}}/2^{\kappa}$.
By a union bound over $q_h$ online guesses, the probability of any successful hit is at most $\min(1,q_h N_{\mathsf{live}}/2^{\kappa})$.
Conditioned on no successful hit, every runtime-visible object is distributed exactly as in the simulator, so any additional distinguishing advantage would imply a plaintext flow outside the declared interface.
\end{proof}

\paragraph{Interpretation.}
Theorem~\ref{thm:sm} is a statement about where plaintext can exist.
It says that the handle-only reference case confines raw secret values to trusted components except through explicit declassification.
It does not, by itself, authorize where a valid handle may be used; those sink restrictions belong to the executor-side authorization path.

The handle-unguessability bound depends on how many live handles the adversary could plausibly hit and, for PRF-generated identifiers, on the PRF loss.

Making this factor explicit avoids overstating confidentiality while still preserving the intended negligible-security conclusion for realistic handle namespace sizes.

\paragraph{Bounded-summary interface.}
\label{app:proofs:bounded-summary}
The deployed benchmark configuration uses a bounded-summary interface.
For a protected value $v$, the runtime receives
\[
  (\mathsf{hid}, D(v))
\]
on an authorized read.
The relevant question is therefore not whether the runtime learns anything about $v$; it does by design.
The question is how much additional distinguishing power the specific read interface introduces beyond base handle leakage.

\paragraph{Advantage definition.}
Let $q_r$ bound the number of authorized protected-read events available to the adversary. Consider a scoped execution trace with at most $m \le q_r$ authorized read events.
For read event $i$, the challenger prepares a pair of secrets $(v_{0,i},v_{1,i})$ with identical base leakage and, for a single hidden bit $b \in \{0,1\}$ shared across the trace, returns one sample from $D(v_{b,i})$.
Let $\mathrm{Adv}^{\mathrm{sm},D}_{\mathcal{A}}(\{(v_{0,i},v_{1,i})\}_{i=1}^{m})$ denote the adversary's left-right distinguishing gap for that trace, i.e., the maximum difference in output-one probability between the two hidden-bit experiments. Under the alternative success-probability convention, the advantage over random guessing is one half of this quantity.

\begin{theorem}[Quantitative exposure of the bounded-summary interface]
\label{thm:bounded-summary}
Assume the same handle conditions and $\varepsilon_{\mathsf{hid}}$ convention as in Theorem~\ref{thm:sm}. Let $D$ be the fixed deployed read interface. Consider a scoped trace in which each authorized read event $i$ emits one sample from $D(v_{b,i})$ using fresh randomness and no secret-dependent cross-read state beyond the emitted summaries themselves. Then for any sequence of read pairs $\{(v_{0,i},v_{1,i})\}_{i=1}^{m}$ with $m \le q_r$ and any PPT adversary making at most $q_h$ handle guesses,

\[
  \mathrm{Adv}^{\mathrm{sm},D}_{\mathcal{A}}(\{(v_{0,i},v_{1,i})\}_{i=1}^{m})
  \le
  \min\!\Bigl(1, \sum_{i=1}^{m} \Delta(D(v_{0,i}),D(v_{1,i}))\Bigr)
  +
  2\left(\min\!\left(1,\frac{q_h N_{\mathsf{live}}}{2^{\kappa}}\right)+\varepsilon_{\mathsf{hid}}\right),
\]
where $\Delta(\cdot,\cdot)$ denotes total variation distance between the per-read output distributions.
\end{theorem}

\begin{proof}
Simulate the handle portion of the two left-right worlds as in Theorem~\ref{thm:sm}. Under the left-right gap convention used here, the corresponding bad-event contribution is at most $2(\min(1,q_h N_{\mathsf{live}}/2^{\kappa})+\varepsilon_{\mathsf{hid}})$.
 
Now hybrid over the $m$ authorized read events.
At hybrid step $i$, switch the $i$th summary sample from $D(v_{0,i})$ to $D(v_{1,i})$ while keeping all other read events fixed.
The distinguishing gap introduced by that single step is at most $\Delta(D(v_{0,i}),D(v_{1,i}))$.
Summing over all $m$ read events yields the first term, truncated at $1$.
Adding the two-world handle-guessing term gives the result.
\end{proof}

The same hybrid argument extends to adaptively scheduled read events by conditioning each step on the preceding transcript and bounding only the increment introduced at the next authorized summary.

\paragraph{Uniform-$\Delta$ special case.}
If every read event in the trace is upper-bounded by the same worst-case distance $\Delta_{\max}$, the theorem reduces immediately to
\[
\mathrm{Adv}^{\mathrm{sm},D}_{\mathcal{A}}
\le
\min(1, q_r \Delta_{\max})
+
2\left(\min\!\left(1,\frac{q_h N_{\mathsf{live}}}{2^{\kappa}}\right)+\varepsilon_{\mathsf{hid}}\right),
\]
which recovers the simpler intuition used informally throughout the paper.

\paragraph{Post-processing cannot increase distinguishability.}
Suppose the deployment implements the visible summary as $D = P \circ D_{\mathrm{core}}$, where $P$ is deterministic or randomized independently of the underlying secret conditional on $D_{\mathrm{core}}(v)$.
Then by the data-processing inequality,
\[
\Delta(D(v_0),D(v_1)) \le \Delta(D_{\mathrm{core}}(v_0),D_{\mathrm{core}}(v_1)).
\]
This is why measuring the deterministic core of the deployed summary operator is the right conservative calibration target, provided the post-processor does not retrieve new secret-bearing context.

\paragraph{Reading the bound.}
If $D$ is constant, then $\Delta = 0$ and the bounded-summary interface collapses to the handle-only reference case.
If $D$ is the identity, then $\Delta = 1$ and the confidentiality term becomes trivial.
Practical deployments lie between those extremes.
The theorem quantifies where and how the bounded-summary interface weakens confidentiality relative to the handle-only reference case on each execution trace.

\begin{table}[h]
\centering
\caption{Qualitative confidentiality behavior of summary-interface families.}
\label{tab:summary-class}
\small
\begin{tabular}{p{0.17\textwidth}p{0.27\textwidth}p{0.10\textwidth}p{0.24\textwidth}}
\toprule
Summary class & Example output & Typical $\Delta$ & Interpretation \\
\midrule
Metadata only & type, size, TTL, sink list & 0 & No content-dependent distinction by construction \\
Structured schema & aliased identifiers, dates, amounts, bounded entity sets & 0 or small & Depends on whether the field is content-carrying or fully symbolized \\
Bounded natural-language summary & short subject/body/list excerpt & up to 1 & Useful for planning but may distinguish secrets directly \\
\bottomrule
\end{tabular}

\end{table}

\paragraph{Measurement protocol for the deployed summary operator.}
The deterministic summary core uses item cap $M{=}8$ per container and character cap $C{=}512$ per text field, before adding truncation indicators. These are not a global bound on the serialized summary length.
We evaluate it on $N=122$ paired inputs spanning 19 secret-location categories representative of the tool outputs that arise in the benchmarked tasks.
Each pair differs in exactly one sensitive location.
Because the measured core is deterministic, each pair induces point distributions and therefore yields $\Delta \in \{0,1\}$ depending on whether the two emitted summaries are byte-identical.
The measurement includes the deployed hardening that resets alias state per read, because without that reset the independence premise of Theorem~\ref{thm:bounded-summary} would be materially weaker. The reported average $\bar{\Delta}=0.287$ below is therefore a diagnostic over this 122-pair measurement suite. It is not a worst-case upper bound for arbitrary execution traces and should not be substituted for the theorem's sum of per-read distances.

\begin{table}[h]
\centering
\caption{Per-category $\Delta$ for the deployed deterministic summary core.}
\label{tab:summary-tv-measured}
\small
\begin{tabular}{p{0.44\textwidth}rr}
\toprule
Secret-location category & Pairs & $\Delta$ \\
\midrule
Identifier and structure-preserving categories \\
\quad Email sender swap & 10 & $0.00$ \\
\quad IBAN counterparty swap & 6 & $0.00$ \\
\quad URL swap & 6 & $0.00$ \\
\quad Domain swap & 4 & $0.00$ \\
\quad Participant list (emails) & 8 & $0.00$ \\
\quad IBAN counterparty, request-mentioned & 6 & $0.00$ \\
\quad Email counterparty, request-mentioned & 6 & $0.00$ \\
\quad Second-read alias index (multi-read turn) & 6 & $0.00$ \\
\quad Canonicalizable email/domain obfuscations & 3 & $0.00$ \\
\midrule
Residual in-window content categories \\
\quad Free-text body within $C=512$ & 12 & $1.00$ \\
\quad Subject line within preview cap & 6 & $1.00$ \\
\quad List items within $M=8$ & 8 & $1.00$ \\
\quad Plain-text transaction counterparty name & 6 & $1.00$ \\
\quad Ambiguously malformed identifier & 3 & $1.00$ \\
\midrule
Beyond-window or stripped-content categories \\
\quad Body content beyond $C=512$ & 8 & $0.00$ \\
\quad List items beyond index $M=8$ & 8 & $0.00$ \\
\quad Dict keys beyond index $M=8$ & 6 & $0.00$ \\
\quad Subject with suspicious markers (stripped) & 4 & $0.00$ \\
\quad Prompt-injection payload inside stripped tag & 6 & $0.00$ \\
\midrule
\textbf{Overall (122 pairs, 19 categories)} & & $\bar\Delta = 0.287$ \\
\textbf{Worst-case category} & & $1.000$ \\
\bottomrule
\end{tabular}

\end{table}

\paragraph{Why end-to-end C2/C5 closure is stronger than the $\Delta$ table alone.}
Even when a read summary is distinguishing, three stronger structural facts continue to hold in the deployed system:
\begin{enumerate}
    \item unauthorized commits still require executor-side authorization and are covered by Theorem~\ref{thm:nbe};
    \item internal channels such as C2 and C5 do not obtain newly dereferenced raw protected values; any content-dependent information they may carry must already originate from the authorized summary interface; and
    \item the read interface itself rate-limits exposure to one bounded summary per authorized read.
\end{enumerate}
The next proposition isolates that point.

\begin{proposition}[Three-layer composition for high-$\Delta$ summary categories]
\label{prop:three-layer}
Let $\mathcal{C}$ be a summary category with $\Delta(D(v_0),D(v_1))=1$.
Under the assumptions of Theorems~\ref{thm:nbe}, \ref{thm:sm}, and \ref{thm:bounded-summary}, an adversary controlling the runtime remains subject to three simultaneous constraints:
\begin{enumerate}
    \item no unauthorized side effect can commit;
    \item the internal relay channels C2 and C5 do not newly expose raw protected values beyond what was already released through the authorized bounded-summary interface; any residual content-dependent signal must already have been released through that interface; and
    \item disclosure through the read path is limited by the authorized, bounded summary interface itself.
\end{enumerate}
Therefore a high-$\Delta$ summary category weakens the per-read confidentiality term of Theorem~\ref{thm:bounded-summary}, but it does not reopen unauthorized commit or a fresh raw-plaintext relay from trusted storage into C2/C5.
\end{proposition}

\begin{proof}
Item 1 follows directly from Theorem~\ref{thm:nbe}, whose premises are independent of the read operator.
Item 2 follows from the interface contract and trust split: C2 and C5 never trigger a fresh dereference of trusted storage; any content-dependent signal observed on those channels must already have been emitted by the authorized summary interface.
Item 3 follows from the read API itself, which emits one bounded summary per authorized read.
These constraints arise from different components and therefore coexist.
\end{proof}

\subsection{Distributed policy evaluation}
\label{app:proofs:policy}

Distributed policy evaluation is optional. Its purpose is to reduce plaintext concentration during policy evaluation. We write $L_{\mathsf{policy}}=(L_{\mathsf{PIR}},L_{\mathsf{MPC}},L_{\mathsf{confirm}},L_{\mathsf{time}})$ for the per-evaluator leakage summarized in Table~\ref{tab:policy-leakage}.
 
\paragraph{Security game.}
We now state the scoped transcript-privacy result for the distributed policy-engine instantiation. The adversary outputs two single actions $x_0,x_1$ such that $L_{\mathsf{policy}}(x_0)=L_{\mathsf{policy}}(x_1)$.
The challenger samples $b \leftarrow \{0,1\}$, executes $x_b$ through the gateway exactly once, and reveals the transcript visible to one policy evaluator $\mathcal{P}_\sigma$.
The adversary wins if it distinguishes which action was executed.
The game is intentionally single-query and non-adaptive because that is the claim this distributed instantiation is designed to support.

\begin{table}[h]
\centering
\caption{Per-evaluator leakage contract for the distributed policy engine.}
\label{tab:policy-leakage}
\small
\begin{tabular}{p{0.18\textwidth}p{0.70\textwidth}}
\toprule
Component & Contents \\
\midrule
$L_{\mathsf{PIR}}$ & PIR endpoint class, bundle or database identifier, padded batch geometry, and public domain size \\
$L_{\mathsf{MPC}}$ & program identifier, public circuit-shape metadata, and batch geometry \\
$L_{\mathsf{confirm}}$ & whether confirmation is required and the coarse confirmation class \\
$L_{\mathsf{time}}$ & coarse timing bucket and fixed-shape schedule metadata \\
\bottomrule
\end{tabular}

\end{table}

\begin{theorem}[Single-query per-evaluator transcript privacy for distributed policy evaluation]
\label{thm:policy-dist}
Assume
(i) the two-server PIR scheme is secure against one server,
(ii) the MPC protocol is secure in the honest-but-curious model,
(iii) the two policy evaluators do not collude, and
(iv) routing is fixed-shape and consistent with the declared leakage function $L_{\mathsf{policy}}$.
Then for each evaluator $\sigma$ there exists a PPT simulator $\mathsf{Sim}_\sigma$ such that the single-query transcript visible to $\mathcal{P}_\sigma$ is computationally indistinguishable from $\mathsf{Sim}_\sigma(L_{\mathsf{policy}})$.
\end{theorem}

\begin{proof}
Use a two-step hybrid argument.
First replace the PIR-visible portion of the transcript with the simulator guaranteed by one-server PIR security.
Then replace the MPC-visible portion with the simulator guaranteed by honest-but-curious MPC privacy.
All remaining transcript fields are explicit functions of $L_{\mathsf{policy}}$ by construction.
Therefore the complete single-evaluator transcript is simulatable from $L_{\mathsf{policy}}$.
\end{proof}

Theorem~\ref{thm:policy-dist} explains why distributing the policy engine can reduce policy-side plaintext concentration without changing the main effect/plaintext split enforced elsewhere in SecureClaw.

\subsection{Trusted-component compromise and guarantee degradation}
\label{app:proofs:compromise}

Theorems~\ref{thm:nbe}, \ref{thm:sm}, and \ref{thm:bounded-summary} are component-scoped: they assume the gateway, handle store, policy engine, and executor are trusted. This section makes explicit what fails if that assumption is violated.

\begin{table}[h]
\centering
\caption{Guarantee degradation under trusted-component compromise.}
\label{tab:compromise-blast-radius}
\small
\begin{tabular}{p{0.28\textwidth}p{0.60\textwidth}}
\toprule
Component compromised & Guarantee impact \\
\midrule
Runtime $\mathcal{R}$ & Still in scope. Request-bound authorization and handle confinement are designed to hold even when $\mathcal{R}$ is adversarial, except for information intentionally released through authorized summaries or policy-allowed actions. \\
Gateway $\mathcal{G}$ & Read confidentiality fails because the gateway can disclose plaintext directly. Canonicalization and preview formation may also become malicious, undermining the semantic meaning of request binding. \\
Handle store $\mathcal{S}$ & Confidentiality of protected plaintext fails. If handle metadata, object bindings, or version tags can be tampered with, dereference integrity may also fail unless separately protected. \\
Policy engine $\mathcal{P}$ & Effect authorization fails if policy incorrectly authorizes forbidden requests. Policy-input confidentiality to that service also fails. \\
One evaluator in the optional distributed variant & Active compromise is outside the honest-but-curious transcript-privacy model. Authorization integrity is no longer guaranteed: a malicious evaluator can submit incorrectly computed shares with valid MACs even when the other evaluator remains honest. \\
Both evaluators in the optional distributed variant & Effect authorization fails. The transcript-privacy claim also becomes void. \\
Executor $\mathcal{X}$ & Effect authorization fails completely because the executor is the sole commit point. Since the executor dereferences handles at commit time, compromise of $\mathcal{X}$ may also expose protected plaintext needed for sink execution. \\
\bottomrule
\end{tabular}

\end{table}

\paragraph{Operational implication.}
SecureClaw materially reduces trust in the runtime, but it does not eliminate trust altogether. High-assurance deployments should therefore pair the architecture with hardening, audit logging, key isolation, and independent monitoring of the gateway, handle store, policy service, and executor.


\section{Additional evaluation}
\label{app:additional-evaluation}

\subsection{Full-inventory ASB cross-model evaluation}
\label{app:full-model-transfer}

We evaluate SecureClaw on the complete ASB inventory with five model snapshots. The comparison includes the primary GPT-4o-mini run and four additional runs. Every run contains the same five attack families with 400 rows each, for 2,000 rows per model.

\begin{table}[htbp]
\centering
\caption{SecureClaw on all 2,000 ASB rows per model snapshot. Utility is benchmark-scored task success. The GPT-4o-mini row is the primary run.}
\label{tab:full-cross-model-asb}
\small
\begin{tabular}{lcc}
\toprule
Model & ASR $\downarrow$ & Utility $\uparrow$ \\
\midrule
GPT-4o-mini (primary) & $0/2000$ & $1778/2000$ (88.9\%) \\
Claude Sonnet 4 & $0/2000$ & $1799/2000$ (90.0\%) \\
GPT-5.6-Luna & $0/2000$ & $1672/2000$ (83.6\%) \\
Qwen3-8B & $0/2000$ & $1883/2000$ (94.2\%) \\
Qwen3.7-Plus & $0/2000$ & $1689/2000$ (84.5\%) \\
\bottomrule
\end{tabular}
\end{table}

No successful attack is observed on any ASB snapshot (Table~\ref{tab:full-cross-model-asb}), while utility ranges from 83.6\% to 94.2\%. The small open-weight Qwen3-8B model has $0/400$ successful attacks in each of the five ASB attack families. Same-snapshot undefended controls have ASRs of 49.8\%, 64.6\%, and 85.4\% for Claude Sonnet 4, GPT-5.6-Luna, and Qwen3.7-Plus, respectively. Their undefended utilities are 31.0\%, 13.8\%, and 1.9\%. Thus, the zero observed ASRs occur on backbones that are susceptible to these attacks without the boundary enforcement.

\subsection{Adapted Progent and CaMeL comparisons}
\label{app:additional-baselines}

Progent~\citep{progent2025} and CaMeL~\citep{camel2025} are evaluated on the same AgentDojo inventory as SecureClaw: 629 attacked rows and 97 benign rows. We use the authors' released implementations through AgentDojo adapters and report the results as \emph{adapted full runs}. This comparison holds the benchmark inventory fixed while accommodating the systems' different execution interfaces. Table~\ref{tab:adapted-agentdojo-full} retains the primary SecureClaw run as its reference.

\begin{table}[htbp]
\centering
\caption{Full-inventory AgentDojo comparison. Progent and CaMeL use adapted released implementations; SecureClaw is the primary run. Attack success and attacked utility use 629 rows; benign utility uses 97 rows.}
\label{tab:adapted-agentdojo-full}
\small
\setlength{\tabcolsep}{5pt}
\begin{tabular}{lccc}
\toprule
System & ASR $\downarrow$ & Attacked utility $\uparrow$ & Benign utility $\uparrow$ \\
\midrule
Progent & $24/629$ (3.82\%) & $337/629$ (53.6\%) & $59/97$ (60.8\%) \\
CaMeL & $0/629$ (0.00\%) & $195/629$ (31.0\%) & $29/97$ (29.9\%) \\
SecureClaw (primary) & $4/629$ (0.64\%) & $356/629$ (56.6\%) & $59/97$ (60.8\%) \\
\bottomrule
\end{tabular}
\end{table}

SecureClaw has 20 fewer observed attack successes than Progent, completes 19 more attacked tasks, and has the same benign task-success count. CaMeL has zero observed attack successes, with lower task utility in both conditions: 195 attacked and 29 benign successes, compared with 356 and 59 for SecureClaw. This comparison exposes a security--utility trade-off between the evaluated configurations. SecureClaw retains the iterative tool-calling loop while achieving low observed ASR; CaMeL reaches the lowest observed ASR in this comparison with a larger utility cost. These counts describe the adapted configurations on this inventory.

\subsection{Protected-field classification and a conservative fallback}
\label{app:field-identification}

The read path relies on identifying the fields that must be protected. We assess this assumption with controlled classification errors, a contextual classification audit, and an end-to-end schema-only fallback. Each evaluation covers the full 996-scenario AgentLeak inventory. For the end-to-end diagnostics in this subsection, a leak is the appearance in $C1$, $C2$, or $C5$ of a concrete value from a field designated as protected by the task's privacy rule. This protected-value check isolates field-identification failures. It is a different measurement from the official parity-lane evaluation that yields the primary $16/496$ any-channel result.

\paragraph{False positives and false negatives.}
With recall held at 1.0, over-marking nonsensitive fields at realized false-positive rates up to 20\% leaves the attacked leak count at $0/496$ and the overall count at $0/996$ at every tested level. Overall utility ranges from 81.6\% to 82.2\%. Over-marking replaces additional values with handles, but its utility cost is small over this tested range.

Randomly missing protected fields has a different effect (Table~\ref{tab:field-misses}). Attacked protected-value leakage increases with the nominal false-negative rate, while overall utility remains at or above 81\%. The gradual measured increase does not remove the recall requirement: a missed field is outside the confinement supplied by its handle replacement.

\begin{table}[htbp]
\centering
\caption{Controlled protected-field misses. Each setting covers 996 scenarios; leakage is the protected-value diagnostic on 496 attacked scenarios, aggregated by scenario across $C1/C2/C5$. Rates in the first column are the nominal random false-negative settings.}
\label{tab:field-misses}
\small
\begin{tabular}{cc}
\toprule
Nominal false-negative rate & Attacked protected-value leaks $\downarrow$ \\
\midrule
1\% & $18/496$ (3.63\%) \\
5\% & $44/496$ (8.87\%) \\
10\% & $65/496$ (13.10\%) \\
20\% & $85/496$ (17.14\%) \\
\bottomrule
\end{tabular}
\end{table}

\paragraph{Contextual classification audit.}
The contextual detector reads the privacy instruction, application domain, record and field names, and declared types. It receives neither field values nor ground-truth field labels, and its decisions are fixed before scoring against those labels. On the development inventory, it identifies 44,448 protected fields and misses 351, for 99.22\% field recall at 86.68\% precision. Every sensitive field is covered on $879/996$ tasks. Field recall and complete task coverage measure different consequences of a miss; the latter requires that no protected field in the task be missed. These are classification results on the development inventory, rather than end-to-end leakage results for this detector.

\paragraph{Schema-only fallback.}
A separate conservative fallback identifies protected fields from the schema without oracle field labels at inference. Its end-to-end run has $0/996$ protected-value leaks on each of $C1$, $C2$, and $C5$, with 85.9\% overall utility. Together with the controlled-error results, this supports favoring recall when configuring field identification: extra handle replacements have limited measured utility impact in the tested range, whereas missed protected fields expose a direct route for plaintext to enter the runtime.

\subsection{Effect and use of benchmark accommodations}
\label{app:accommodation-study}

Alias resolution maps reference IDs to the task-relevant literals expected by benchmark tools and exact-string utility scoring. Auto-confirmation supplies a confirmation event in non-interactive runners when an already policy-allowed request requires one. We evaluate their effects by running the complete AgentDojo inventory under four conditions (Table~\ref{tab:accommodations-full}). Each condition is a separate hosted run on the same 629 attacked and 97 benign rows. These runs supplement the primary result of $4/629$ ASR and do not replace it.

\begin{table}[htbp]
\centering
\caption{AgentDojo accommodation study. Each arm contains 629 attacked and 97 benign rows. The full arm enables both alias resolution and auto-confirmation.}
\label{tab:accommodations-full}
\small
\setlength{\tabcolsep}{4pt}
\begin{tabular}{lccc}
\toprule
Condition & ASR $\downarrow$ & Attacked utility $\uparrow$ & Benign utility $\uparrow$ \\
\midrule
Full & $1/629$ (0.16\%) & $331/629$ (52.6\%) & $55/97$ (56.7\%) \\
No alias resolution & $0/629$ (0.00\%) & $282/629$ (44.8\%) & $45/97$ (46.4\%) \\
No auto-confirmation & $2/629$ (0.32\%) & $341/629$ (54.2\%) & $56/97$ (57.7\%) \\
Neither & $0/629$ (0.00\%) & $294/629$ (46.7\%) & $47/97$ (48.5\%) \\
\bottomrule
\end{tabular}
\end{table}

\paragraph{Alias use and task scoring.}
In the full arm, downstream calls resolve aliases on $129/726$ rows: $112/629$ attacked rows and $17/97$ benign rows, with 216 resolution events. Comparing the same task rows across the full and no-alias conditions, 65 attacked rows succeed only in the full condition and 16 succeed only without aliases. The corresponding benign counts are 12 and 2. The attacked difference of 49 rows equals the aggregate utility gap; the 65 full-only successes constitute 19.6\% of the 331 full-arm attacked successes. Because the conditions use separate hosted executions, these discordant outcomes describe a row-level comparison rather than deterministic attribution of every success to alias resolution.

Trusted code resolves a token such as \texttt{EMAIL\_REF\_3} to the benchmark-expected literal. This permits a correct reference ID to satisfy the scorer, while the resolved recipient, channel, and context remain subject to policy authorization and exact-request binding. AgentLeak's leak detectors score raw outputs, so its primary $16/496$ leakage result does not receive this alias-resolution accommodation. ASB uses its native attack scorer; its separate recovery ablation is reported in Appendix~\ref{app:recovery-ablation}.

\paragraph{Confirmation use and dependence.}
No confirmation token is issued in any AgentDojo arm ($0/726$ rows per arm). This includes all policy-denied proposals: $0/43$, $0/82$, $0/38$, and $0/62$ in the full, no-alias, no-confirmation, and neither conditions, respectively. Consequently, none of the 331 attacked task successes in the full arm uses auto-confirmation. Auto-confirmation is also never invoked in the primary ASB run, so none of its 1,778 scored-successful rows depends on it.

AgentLeak's non-interactive adapter supplies confirmation at final-output release. A matched no-confirmation counterfactual is unavailable for the primary AgentLeak configuration, so its utility dependence on this accommodation is not quantified here. In all cases, confirmation is applied after policy authorization and cannot turn a policy-denied request into an allowed release or commit.

\end{document}